\newif\ifarxiv
\arxivtrue 
\ifarxiv \documentclass[letterpaper]{article}
\usepackage[margin=1in]{geometry}
\usepackage{natbib} 
\PassOptionsToPackage{hyphens}{url}\usepackage{hyperref}
\usepackage{graphicx}
\else 
\documentclass[letterpaper]{article} 
\usepackage[submission]{aaai25}  
\usepackage{times}  
\usepackage{helvet}  
\usepackage{courier}  
\usepackage[hyphens]{url}  
\usepackage{graphicx} 
\usepackage{natbib}  
\usepackage{caption} 
\fi 

\usepackage{amsmath, amsfonts}
\usepackage{amsthm}  
\usepackage{thmtools,thm-restate}
\usepackage{amssymb} 
\usepackage{thmtools,thm-restate}
\usepackage[autostyle=true,german=quotes]{csquotes}

\usepackage[ruled]{algorithm2e} 

\newcommand{\p}[1]{\left( #1 \right)}

\newcommand{\cd}[0]{\cdot}
\newcommand{\abs}[1]{\left \vert #1 \right \vert}

\newtheorem{theorem}{Theorem}

\newtheorem{definition}{Definition}

\newtheorem{hypothesis}{Hypothesis}

\newcommand{\nitem}[0]{\ensuremath{M}}
\newcommand{\nplayer}[0]{\ensuremath{N}}

\newcommand{\weight}[0]{\ensuremath{w}}
\newcommand{\costun}[0]{\ensuremath{c}}
\newcommand{\ntype}[0]{\ensuremath{T}}
\newcommand{\type}[0]{\ensuremath{t}} 
\newcommand{\bias}[0]{\ensuremath{\beta}}

\newcommand{\omt}[1]{}
\newif\ifinformal
\informalfalse

\usepackage{color}

\ifarxiv

\begin{document} 

\title{Private Blotto: Viewpoint Competition with Polarized Agents
}

\author{Kate Donahue\thanks{Department of Computer Science, Cornell University} \and Jon Kleinberg\thanks{Departments of Information and Computer Science, Cornell University}}
\date{} 
\else 
\title{Private Blotto: Viewpoint Competition with Polarized Agents}
\author{
    Written by AAAI Press Staff\textsuperscript{\rm 1}\thanks{With help from the AAAI Publications Committee.}\\
    AAAI Style Contributions by Pater Patel Schneider,
    Sunil Issar,\\
    J. Scott Penberthy,
    George Ferguson,
    Hans Guesgen,
    Francisco Cruz\equalcontrib,
    Marc Pujol-Gonzalez\equalcontrib
}
\affiliations{
    \textsuperscript{\rm 1}Association for the Advancement of Artificial Intelligence\\


    1101 Pennsylvania Ave, NW Suite 300\\
    Washington, DC 20004 USA\\
    proceedings-questions@aaai.org
%
}

\begin{document}
\fi

\maketitle

\begin{abstract} 
Social media platforms are responsible for collecting and disseminating vast quantities of content. Recently, however, they have also begun enlisting users in helping annotate this content - for example, to provide context or label disinformation. However, users may act strategically, sometimes reflecting biases (e.g. political) about the \enquote{right} label.  How can social media platforms design their systems to use human time most efficiently? Historically, competition over multiple items has been explored in the Colonel Blotto game setting\cite{borel1921theorie}. However, they were originally designed to model two centrally-controlled armies competing over zero-sum \enquote{items}, a specific scenario with limited modern-day application. In this work, we propose and study the Private Blotto game, a variant with the key difference that individual agents act independently, without being coordinated by a central \enquote{Colonel}. We completely characterize the Nash stability of this game and how this impacts the amount of \enquote{misallocated effort} of users on unimportant items. We show that the outcome function (aggregating multiple labels on a single item) has a critical impact, and specifically contrast a majority rule outcome (the median) as compared to a smoother outcome function (mean). In general, for median outcomes we show that instances without stable arrangements only occur for relatively few numbers of agents, but stable arrangements may have very high levels of misallocated effort. For mean outcome functions, we show that unstable arrangements can occur even for arbitrarily large numbers of agents, but when stable arrangements exist, they always have low misallocated effort. We conclude by discussing implications our results have for motivating examples in social media platforms and political competition. 

\end{abstract}

%

\section{Introduction}
Over the last several decades, social media platforms have become hubs of information, responsible for collecting and disseminating vast quantities of content. The sheer scale of content means that traditional sources of annotation and curation (e.g. traditional fact-checking) has become borderline infeasible, even as these platforms have become primary sources of information for many people 
 (e.g. see \cite{doi:10.1080/21670811.2013.872420, 10.1111/j.1083-6101.2012.01574.x}). Instead, some of these platforms have turned to other solutions, including using the platforms and users themselves to curate and annotate content. For example, the Community Notes tool on X.com (formerly known as the Birdwatch tool on Twitter \cite{wojcik2022birdwatch}) has the goal of identifying misinformation by allowing X users to vote on \enquote{notes} with added context that are used to annotate posts. 

In this work, we will draw on game-theoretic tools to help analyze scenarios like this, where multiple strategic, biased agents compete over several different items\footnote{Other settings where this occurs include political competition over multiple issues: see Section \ref{sec:examples} for a discussion.}. In particular, over a century ago, before modern game theory was fully established, 
\'Emile Borel proposed a family of related zero-sum games that study a similar problem to ours, in the \emph{centralized} setting. Specifically, this models centrally-coordinated competition using military conflict as a metaphor \cite{borel1921theorie}:

\begin{definition}[Colonel Blotto]
Two players, $A$ and $B$, are competing over $\nitem$ different
\emph{fronts}, with $\nplayer_a, \nplayer_b$ units of
effort at their disposal respectively. A player wins a front if
they allocate more effort to the front than their opponent does,
and each player wishes to win as many
fronts as possible. Are there Nash stable arrangements of effort over
fronts, and if so, which are they?  
\end{definition}

The name ``Colonel Blotto''
comes from the fact that a colonel controls multiple
individual soldiers, which they allocate across the
battlefields in order to serve their overall objective. The Colonel
Blotto game has been the focus of extensive exploration, including
variants that allow for battlefields to have different values, for
effort to be allocated probabalistically, and for smoother utility
functions \cite{golman2009general, hart2008discrete, OSORIO2013164} (see Appendix \ref{app:blottorelated} for more works). One common thread has been the relative scarcity of pure Nash equilibria, which has centered the literature around exploration of mixed Nash equilibria. The game has also found many applications in areas far removed from warfare
\cite{merolla2005play}, such as national politics, but always with centralized entities competing over multiple items. 

\paragraph{\bf Modeling decentralized conflicts.}
However, the Colonel Blotto framing is at odds with the modern type of disaggregated competition, such as our motivating example of social media users labeling items where multiple agents may share similar goals, but are not controlled by a central organizing \enquote{Colonel}. How might we model this type of political {\em viewpoint competition}?

We could imagine that there is a large collection of agents (e.g.  users), each of
whom is interested in taking part in a conflict with $\nitem$ ``items'' (e.g. online posts)
Each agent controls only one unit of effort, and can choose to devote
that effort to one of the items (i.e. labeling a post as misinformation or not).
There is no centralized ``colonel'' to direct the agents, but instead each agent $i$ has one of two {\em types} (represented by real numbers $\beta_a, \beta_b$), which we can think of as
a viewpoint, bias, or political position. 
After each agent chooses an item to engage with, the outcome of
the conflict on each given item is determined by an {\em outcome function}
that takes the multiset of types at that item and determines a
real-valued outcome. We will also include a positive penalty $\costun$ for leaving an item unlabeled (empty).

An \enquote{outcome function} determines how inputs from multiple users results in a single label for the item. By selecting one outcome function over another, the designer could influence how individual agents choose to exert their effort. Thus, the choice of outcome function will be one of the central focus points of this paper -- it represents one of the few aspects that the designer of the system may have control over (where the designer could be a social media company or political entity running the election, for example).  

We will be particularly interested in two natural
outcome functions for aggregating viewpoints in this setting: 
the {\em median} (in which 
the outcome on an item is the median of the types there) and the
{\em mean} (in which the outcome is the mean of the types).
Agents want the outcomes on each item (even the ones
where they don't participate) to match their types;
thus, each agent experiences a cost equal to the average of 
the distances between the outcome on each item and the agent's type.

We will refer to this type of game as {\em Private Blotto};
like Colonel Blotto, it involves conflict over multiple items,
but it is fundamentally different because it is designed to
model decentralized conflict where each individual agent makes
their own choice about which item 
to participate in\footnote{In the military,
a \emph{private} is an enlisted soldier at the base of the hierarchy.
This reflects our setting, which views the individual
soldiers as the strategic actors, rather than the coordinating colonel
who commands the army.}.
We summarize the discussion above in the following definition.

\begin{definition}[Private Blotto]
Two types of agents, type $A$ and $B$, are competing over $\nitem$ different items, with $\nplayer_a$ agents of type $A$ and $\nplayer_b$ agents of type $B$. Each agent chooses exactly one item to compete in (label), and an outcome function (for example, the median or mean) determines the outcome value on each item. An agent's cost is equal to the average distance between the outcome on each item and the agent's type. Are there Nash stable arrangements of agents over items, and if so, what do they look like?
\end{definition}

For this class of games, we can explore a number of questions. One of the most fundamental contrasts we will show is that the decentralized \emph{Private} Blotto game admits pure Nash equilibria more frequently than the centralized \emph{Colonel} Blotto game: part of our work's contribution will be to characterize when these equilibria occur, and what they look like. Additionally, one vital question is how the choice of outcome function affects the existence of stable arrangements, and how those stable arrangements distribute agents across items. Finally, we will explore how stable arrangements compare according to normative goals of utilizing user time well (minimizing \emph{misallocated effort} of users onto unimportant items). 
We now provide some more detail on settings that can be
modeled by the Private Blotto game, and then we give an
overview of our results.

\subsection{Motivating examples and further related work}\label{sec:examples}

Our Private Blotto formulation finds applicability in numerous modern-day settings. Here, we will describe a few key application areas in more detail. 

Crowdsourcing on social media has become a growing area of societal and academic interest in recent years \cite{yasseri2021can, wojcik2022birdwatch, birdsdontfactcheck, prollochs2022community}. Focusing on Community Notes, users can provide (discrete) labels on tweets, labels which have been shown \cite{birdsdontfactcheck} to have partisan bias.  Empirical studies \cite{birdsdontfactcheck, saeed2022crowdsourced} show most tweets only have 1-2 labels and the modal user submits only one note (median user submits 5), which matches the Private Blotto setting of  discrete labels with bandwidth-limited agents. Other aspects of Community Notes match Private Blotto well: users are given pseudonyms when voting, making coordination between users implausible, and typically see the labels on each tweet before deciding to label it, mirroring how users in Private Blotto determine how to allocate effort based on the existing set of labels \cite{wojcik2022birdwatch}. Because Community Notes users are also X users, they view tweets that aren't already labeled, and may incur some disutility for leaving misinformation unlabeled, motivating Private Blotto's positive cost $\costun$ for leaving an item empty. We describe other related papers in crowdsourcing in Appendix \ref{app:crowdrelated}.

 Separately, our Private Blotto setting also finds applications in political contests or issue-based activism, as well as military engagements, which have both historically been application areas for Colonel Blotto \cite{merolla2005play}. For political contests, we argue that Private Blotto might even be a more natural fit. Here, the $\nitem$ items might represent issues or political campaigns, while the agents might be activist groups or donors, which might share similar goals, but are unable (for logistical or legal reasons) to coordinate their actions. Differing types would reflect differing political leanings, which could be closer or further apart (reflecting the magnitude of the gap in biases). For military engagements, in Private Blotto each agent might be an individual soldier, guerrilla member, or other actor that is acting without coordination from some central organizer. In this way, Private Blotto might naturally model more modern types of asymmetric warfare conflicts, where agents on the same \enquote{side} militarily are of the same type. 

\subsection{Overview of results}

We are primarily interested in which instances of Private Blotto games admit \emph{stable arrangements} where no agents has an incentive to unilaterally change items. We are also curious about the properties of stable arrangements, when they exist -- do they result in agents being spread out across items, or could they involve many agents clustered on a single item?

We will find it useful to divide our analysis of the model into two
main cases, depending on whether there are more agents than items (Section \ref{sec:moreagents}) or more items than agents (Section \ref{sec:feweragent}).

When there are more agents (Section \ref{sec:moreagents}), we will find that the choice of outcome function (mean or median) can produce starkly different results. We can view these results as considering the plane of ($\nplayer_a, \nplayer_b$) pairs and asking whether a stable arrangement must exist for $\nplayer_a$ type $A$ players and $\nplayer_b$ type $B$ players. In the case of the median outcome function, we show that there is a \emph{median-critical region} of bounded size where stable arrangements fail to exist: in particular, this means that given sufficiently many agents, a stable arrangement is always guaranteed to exist. However, these stable arrangements could have almost all agents clustered on a single item. If we take the normative principle that items (representing posts or political issues) should all receive approximately equal levels of attention, such disproportionate levels of agents could be viewed as a high level of \emph{misallocated effort} on the part of agents. 

By contrast, the mean outcome function results in very different results for stable arrangements. In particular, we show that there are arbitrarily many (and arbitrarily large) $(\nplayer_a, \nplayer_b)$ pairs where no stable arrangement exists, showing that stability may be much harder to guarantee. However, we also show that when a stable arrangement exists, it must have all agents split almost exactly evenly across items (up to integer rounding), resulting in a very low level of misallocated effort. 

In Section \ref{sec:feweragent} we turn to the scenario where there are more items than agents: intuitively, these are settings where some items will inevitably be left empty. Here, we show that median and mean outcome functions produce very similar outcomes, which is natural given that the mean and median are identical for small numbers of items. At a high level, while we show that while settings without any stable arrangements can frequently exist, so long as the cost for leaving an item empty is sufficiently high, there is always a stable arrangement where players spread out with exactly one player per item. 

Finally, Section \ref{sec:discussion} concludes and discusses implications that our results may have for our motivating examples in post annotation in social media and political competition over issues. In particular, our results show that the choice of outcome function can dramatically influence how biased agents may choose to expend their effort.

\nocite{borel1921theorie, ahmadinejad2019duels, golman2009general, hart2008discrete, OSORIO2013164, OSORIO2013164, skaperdas1996contest, attackdefense, schwartz2014heterogeneous, kovenock2012coalitional, boix2020multiplayer, anbarci2020proportional, mazur2017partial, hettiachchi2022survey, zhang2017consensus, 6195597, yasseri2021can, wojcik2022birdwatch, birdsdontfactcheck, prollochs2022community, saeed2022crowdsourced, Behnezhad_Dehghani_Derakhshan_Hajiaghayi_Seddighin_2017, adam2021double}

\section{Model}\label{sec:model}

In this section, we make our theoretical model more precise. We assume there are $\nitem$ \emph{items} that $\nplayer$ total \emph{agents} are competing over. Each agent controls exactly 1 unit of effort: they may choose which item to compete in, but may not coordinate with other players. However, agents come in two \emph{types} ($A$ and $B$). Two agents of the same type have perfectly aligned incentives: when present on the same item, they work towards the same outcome, and when on different items, two agents of the same type are interchangeable. Each type has a real-valued \emph{bias} $\beta_a, \beta_b \in \mathbb{R}$ that describes how similar or dissimilar the types are to each other (how polarized the two groups are). For example, $\beta_a = 1, \beta_b = -1$ agents are closer to each other than $\beta_a = 5, \beta_b = -3$.

\subsection{Outcome functions}
Once agents are arrayed on an item, the outcome of the battle is governed by an \emph{outcome function} $f(\cd)$. In this paper, we will focus on two types of outcome functions: \emph{median} outcome and \emph{mean} (proportional) outcome. Given a set of agents $S_i$ on item $i$, the median outcome function returns the median of the biases $\text{med}(\{\bias_{\type}\} \  \vert \  \type \in S_i) $. 
If there are an even number of players on a particular item, then the median function averages together the middle two biases. Given two types of players, the median outcome function is equivalent to \enquote{winner-take-all}, where whichever type dominates the item wins. 

On the other hand, the mean outcome function returns the mean of the biases: 
$\frac{1}{\abs{S_i}}\sum_{\type \in S_i}\bias_{\type}$. This models a scenario where the final outcome of the item depends on the distribution of agent biases, not solely the median agent.

For both mean and median outcome function, we assume agents have cost given by the distance of their bias to the outcome: that is, given outcome function $f(\cd)$ on an item with set of labels $S_i$, then an agent with bias $\bias_t$ has cost for that item of $\abs{f(S_i) - \bias_t}$. 

\subsection{Agent cost}
Even though each agent only participates in a single item, we model the agents as still having preferences over all of the items. This could reflect settings where the agents are social media users who observe multiple posts but only have the bandwidth to provide annotations on a smaller subset, or political actors who have opinions about many topics but only focus their energy on a single issue. Additionally, we assume each agent experiences a cost $\costun\geq 0$ for leaving the item empty, which is independent of agent bias. We include this feature to model settings where agents may choose to leave an item empty (to join a more contested item), but suffer some non-zero cost in doing so.  We can write the total cost as: 
$$\sum_{i \in [\nitem], \abs{S_i}>0} \abs{f(S_i) - \bias_{\type}} + \sum_{i \in [\nitem], \abs{S_i} =0} \costun$$

We will say that an arrangement of agents across items is \emph{stable} if it satisfies Nash stability (no agent can unilaterally decrease its cost). In the online content labeling setting, if an arrangement fails to be Nash stable, this means that some online users would prefer to move which posts they label, or potentially generate more labels --- which could cause unstable cycles where users compete in an arms race to generate more posts.

\begin{definition}[Nash stability (pure)]\label{def:nash}
An arrangement of players on items is (Nash) stable in the Private Blotto game if no agent can reduce its cost by switching from competing in one item to begin competing in another.

\end{definition}

Inherent in this definition of stability is the requirement that agents are \emph{decentralized}: in particular, each agent is deciding which action to take by themselves, without coordinating among other agents of the same (or similar) biases. This definition is a departure from the prior Colonel Blotto literature which allowed multiple agents to be coordinated by a central \enquote{Colonel}. However, Definition \ref{def:nash} is the natural definition of stability to study for more decentralized settings (e.g. crowdsourcing, disaggregated political contests) with self-interested actors. Note that we focus on the most natural question of \emph{pure} Nash equilibria, although exploring mixed Nash equilibria would be an interesting extension for future work.

\section{More agents: $\nplayer \geq \nitem$}\label{sec:moreagents}

In analyzing the Private Blotto game, we will find it helpful to divide our analysis into two main regimes: when there are more agents than items (this section), and when there are fewer agents than items (Section \ref{sec:feweragent}). As related to the examples in Section \ref{sec:examples}, in online crowdsourcing this models scenarios where there are enough online users that they could provide each post with at least one label, and in political issues it could represent the case where there is a relatively small subset of major divisive issues that multiple political actors are debating. Without loss of generality, we will always name the two types so that $\nplayer_a \geq \nplayer_b$ (there are more type $A$ than type $B$ players).

Lemma \ref{lem:nounlabeled} shows the setting with more agents than items is especially clean: so long as the cost for leaving an item empty is sufficiently high, then no item will be left empty, and all results will be completely independent of agent biases $\bias_i$.

\begin{restatable}{lemma}{nounlabeled}
\label{lem:nounlabeled}
If there are more agents than items ($\nplayer \geq \nitem$) and the cost for leaving a item empty is sufficiently high, then no item will be left empty, regardless of if median or mean outcome function is used. Specifically, this occurs when:  
 $$\costun\geq 0.5 \cd \abs{\bias_a - \bias_{b}}$$ 
Moreover, agent strategy becomes independent of biases $\bias_a, \bias_b$ and relies solely on the number of agents of each type on each item, $\{a_i, b_i\}, \ i \in [\nitem]$. 
\end{restatable}
Proofs for Lemma \ref{lem:nounlabeled}, as well as for other proofs in this paper, are found in Appendix \ref{app:proofs}. For the rest of this section, unless stated otherwise, we will assume that the preconditions of Lemma \ref{lem:nounlabeled} hold, which  ensures that no item will be left empty. This assumption is mainly made for cleanness of analysis: if it is relaxed, then the value of $\costun$ causes minor changes in the stable arrangements, primarily for small numbers of agents $\nplayer$.

\subsection{Median outcome}\label{sec:moreagentsmedian}

First, in this section we explore Private Blotto with $\nplayer\geq \nitem$ with median outcome function, modeling the case where contests are decided by a winner-take-all outcome. We can view this setting as exploring the $\nplayer_a, \nplayer_b$ plane, studying for which values of $\nplayer_a, \nplayer_b$ a Nash equilibrium exists, as well as constructively producing an example of a stable arrangement. Our results will be a function of the total number of items ($\nitem$), as well as $\nplayer_a, \nplayer_b$, the number of agents of types $A$ and $B$ respectively. (Recall that given Lemma \ref{lem:nounlabeled} all of our results will be independent of the biases $\bias_a, \bias_b$. )

Our main result for this section is Theorem \ref{thrm:medianmoreagents}, which exactly characterizes when a stable arrangement exists for median outcome. Specifically, this occurs whenever the number of agents $\nplayer_a, \nplayer_b$ for each of the types is \emph{not} in the median-critical region (Definition \ref{def:mediancritical}). 

\begin{definition}[Median-critical region]\label{def:mediancritical}
A set of parameters $(\nplayer_a, \nplayer_b)$ is in the median-critical region if they satisfy: 
$$\nplayer_a + \nplayer_b \leq 2 \cd \nitem \text{ and }  \nitem < \nplayer_a \text{ and  } 1 \leq \nplayer_b < \nplayer_a - \nitem$$
and symmetrically if the roles of $\nplayer_a, \nplayer_b$ are reversed. 
\end{definition}

\begin{theorem}\label{thrm:medianmoreagents}
Given more agents than items ($\nplayer_a + \nplayer_b\geq \nitem$), there exists a stable arrangement if and only if $(\nplayer_a, \nplayer_b)$ is \emph{not} in the median-critical region. 
\end{theorem}

We will prove Theorem \ref{thrm:medianmoreagents} through several sub-lemmas which collectively handle different values of $(\nplayer_a, \nplayer_b)$. First, Lemma \ref{lem:medianneverstable} proves that any set of parameters within the median-critical region must always result in an unstable arrangement. 

\begin{restatable}{lemma}{medianneverstable}
\label{lem:medianneverstable}
Given median outcome and cost satisfying Lemma \ref{lem:nounlabeled}, for any set of biases $\beta_a, \beta_b$, for all instances within the median-critical region, there is never a stable arrangement of agents onto items.  
\end{restatable}

Next, we will prove that any number of agents $(\nplayer_a, \nplayer_b)$ outside of the median-critical region must have a stable arrangement. Lemma \ref{lem:median2mstable} handles the case where there are more than twice as many agents as there are items, showing that this implies there must always exist a stable arrangement.

\begin{restatable}{lemma}{medianmstable}\label{lem:median2mstable}
Given median outcome and cost satisfying Lemma \ref{lem:nounlabeled}, if $\nplayer_a + \nplayer_b \geq 2 \cd \nitem+1$ (or $\nplayer_a + \nplayer_b =2 \cd \nitem$ with $\nplayer_a, \nplayer_b$ even), then there always exists a stable arrangement. 
\end{restatable}

Finally, we address the question of $(\nplayer_a, \nplayer_b)$ pairs with $\nplayer_a + \nplayer_b \leq 2 \cd \nitem$ (so that they are not addressed by Lemma \ref{lem:median2mstable}), but which also do not fall in the median-critical region (so they are not addressed by Lemma \ref{lem:medianneverstable}). Lemma \ref{lem:medianstableequal} examines this case and constructively shows that it is always possible to find a stable arrangement of agents onto items.

\begin{restatable}{lemma}{medianstableequal}
\label{lem:medianstableequal}
Given median outcome and cost satisfying Lemma \ref{lem:nounlabeled}, any number of agents ($\nplayer_a, \nplayer_b$) with $\nplayer_a + \nplayer_b \leq 2 \cd \nitem$ (besides those in the median-critical region) always has a stable arrangement. 
\end{restatable}
Taken together, these lemmas prove Theorem \ref{thrm:medianmoreagents}, exactly characterizing when a stable arrangement exists in Private Blotto games. In particular, they show that unstable arrangements are relatively \emph{rare}, constrained only to the median-critical region (Definition \ref{def:mediancritical}). For small $\nitem$, this region can be very small. For example, it is empty for $\nitem=2$, implying that there is always a stable arrangement for median outcome function with two items. For $\nitem=3$, the median-critical region contains only two points ($\nplayer_a = 4, \nplayer_b =1$), ($\nplayer_a = 5, \nplayer_b =1$), showing that for almost all ($\nplayer_a, \nplayer_b$) pairs, a stable arrangement must exist.

\subsection{Mean outcome}\label{sec:moreagentsmean}

Next, in this section we will explore the setting where again there are more agents than items ($\nplayer\geq\nitem$), but where instead the mean outcome function is used. At a high level, in Section \ref{sec:moreagentsmedian} we proved that there were only a finite number of pairs ($\nplayer_a, \nplayer_b$) such that no stable arrangement of players onto items existed. By contrast, in this section we will show that for mean outcome function, there are arbitrarily many pairs ($\nplayer_a, \nplayer_b$) with no stable arrangement. 

For illustrative purposes, Figure \ref{fig:stable} numerically explores when a stable arrangement exists for $\nitem =2$ items\ifarxiv \footnote{Code to reproduce figures and numerical examples is available at \url{https://github.com/kpdonahue/private_blotto}.} \else \footnote{Code to reproduce figures and numerical examples will be provided in the final version.} \fi. The axes represent the total number of type $A$ and type $B$ agents that are present, with a red dot appearing at point $(\nplayer_a,\nplayer_b)$ if no possible stable arrangement involving that number of agents exists. Note that the red dots extend for high values of $\nplayer_a, \nplayer_b$, indicating that stable arrangements can fail to exist even for large numbers of players. This is in stark contrast to the median function in Section \ref{sec:moreagentsmedian}, where we showed that a stable arrangement only fails to exist within the (bounded) median-critical region. For example, the corresponding plot to Figure \ref{fig:stable} for median would not have any red dots at all (given that the median-critical region is empty for $\nitem=2$ items).

\begin{figure}
    \centering
    \includegraphics[width=3in]{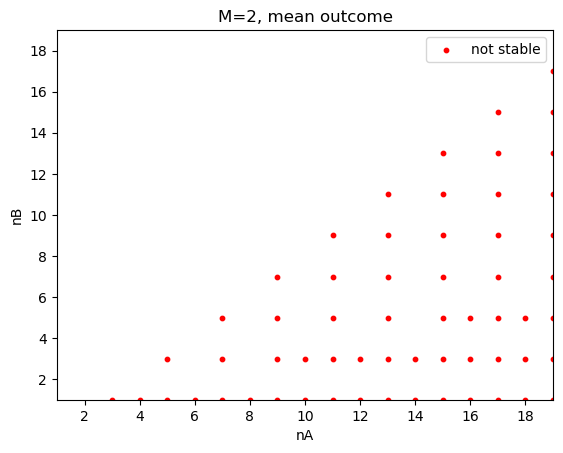}
    \caption{Figure illustrating when stable arrangements of agents onto items fail to exist for mean outcome ($\nitem=2$ items). For clarity, only displayed for $\nplayer_a \geq \nplayer_b$. }
    \label{fig:stable}
\end{figure}
Theorem \ref{thrm:meanbigunstable} formalizes the intuition from Figure \ref{fig:stable}, proving that for any arbitrarily large set of players $\nplayer_a, \nplayer_b$, it is possible to find $\nplayer_a'\geq \nplayer_a, \nplayer_b' \geq \nplayer_b$ such that no stable arrangement of players onto items exists when mean outcome is used. In particular, the construction within Theorem \ref{thrm:meanbigunstable} involves $\nplayer_a = \nplayer_b+2$, with both $\nplayer_a, \nplayer_b$ odd: note that in Figure \ref{fig:stable}, all such points (odd numbers exactly 2 apart) have red dots, indicating no stable arrangement exists. 

\begin{restatable}{theorem}{meanbigunstable}\label{thrm:meanbigunstable}
For every $\nplayer_a\geq \nplayer_b$, there exists an $\nplayer_a' \geq \nplayer_a, \nplayer_b' \geq \nplayer_b$ and $\nitem$ such that there is no stable arrangement of $\nplayer_a', \nplayer_b'$ players onto $\nitem$ items.
\end{restatable}

What is driving this persistent pattern of instability? Consider a hypothetical variant of Private Blotto where agents can allocate effort fractionally across items\footnote{In this setting player payoff is similarly given by the mean outcome, but with $a_i, b_i \in \mathbb{R}$: see Appendix \ref{app:proofs} for details.}. Theorem \ref{thrm:meanbowl} shows that this would always lead to a stable arrangement exactly equal to even allocation over items. Thus, persistent recurrence of instability in Figure \ref{fig:stable} are driven solely by the requirement that effort be allocated in integer units.

\begin{restatable}{theorem}{meanbowl}
\label{thrm:meanbowl}
For $\nitem$ items with two types of agents, $A$ and $B$ with mean outcome and $\costun$ satisfying the conditions of Lemma \ref{lem:nounlabeled}, if players are allowed to be allocated fractionally over items, then the stable arrangement is always given by $a_i = \nplayer_a/\nitem, b_i = \nplayer_b/\nitem$. 
\end{restatable}

Note that, in general, Theorem \ref{thrm:meanbowl} does \emph{not} imply that stable arrangement for the integer-valued Private Blotto games will be close to proportional. While Theorem \ref{thrm:meanbowl} can be extended to show that the fractional Private Blotto game is convex, it is known that in general, the minimum of a convex function, when restricted to integer values may be arbitrarily far from the minimum of the same convex function over real numbers\footnote{For an illustrative example, see \cite{stackoverflow} cited in the references.}. However, it turns out that for the Private Blotto game it \emph{is} true that integer-valued stable arrangements are \enquote{close} to proportional. This idea is formalized in Theorem \ref{thrm:meanclose}: 

\begin{restatable}{theorem}{meanclose}
\label{thrm:meanclose}
Given mean outcome function, any arrangement that is stable must be \enquote{close} to proportional: $\abs{a_i - \nplayer_a/\nitem} \leq 1, \abs{b_i - \nplayer_b/\nitem} \leq 1$ for $i \in [\nitem]$, given $\costun$ satisfying the conditions of Lemma \ref{lem:nounlabeled}.
\end{restatable}

\subsection{Misallocated effort}\label{sec:misallocate}
Finally, in this section we will compare the stable arrangements given either median or mean outcome functions. In particular, we will consider the question of how \enquote{bad} stable arrangements might be, which may influence which type of outcome function might be most desirable for a given contest. This question has been formalized in a variety of ways in previous papers, including Price of Anarchy or Price of Stability \cite{koutsoupias2009worst, anshelevich2008price}. For example, in a congestion routing game, Price of Anarchy would measure the total congestion for all players in the worst-case stable arrangement, as compared to the arrangement that minimizes total congestion. 

However, Private Blotto is modeling a fundamentally different game. In the plaform annotation scenario, Private Blotto is modeling different online users competing over posts for which they have truly different viewpoints, influenced by their personal biases and knowledge. In the absence of impartial, \enquote{gold standard} misinformation investigation (which may be impossible to do at scale), the social media company (and society at large) may not have a clear sense of which misinformation label is \enquote{right}. However, we may have a normative preference that all posts should obtain roughly equal numbers of labels. For example, if one post has dozens of battling annotations while other posts go unlabeled, we might view that allocation of human effort as wasteful\footnote{There are multiple natural extensions to the Private Blotto game, including cases where some items are more important than others, or where players come in arbitrarily many types: see Appendix \ref{app:unevenweight} for a discussion and extension of our results.}.  This intuition of \enquote{misallocated effort} is formalized below:  

\begin{definition}\label{def:misallocate}
Given an arrangement of agents onto items, we say it has \enquote{misallocated effort} given by the amount of agents that is above or below equal allocation. That is, misallocated effort is given by: 
$$\sum_{i \in [\nitem]}\abs{\frac{\nplayer_a}{\nitem} - a_i} + \abs{\frac{\nplayer_b}{\nitem} - b_i}$$ 
\end{definition}

One question we will explore is the maximum possible misallocated effort, among any stable arrangements. Here, we will show a qualitative difference in the bound depending on whether mean or median outcome functions are used.

First, Lemma \ref{lem:meanmisallocate} shows that misallocated effort is upper bounded by a constant times the number of items $\nitem$, driven by the results of Theorem \ref{thrm:meanclose}.

\begin{restatable}{lemma}{meanmisallocate}\label{lem:meanmisallocate}
For mean outcome function,  misallocated effort is upper bounded by $2 \cd \nitem$. 
\end{restatable}

Next, Lemma \ref{lem:medianmisallocate} shows that worst-case misallocated effort can be much higher for the median outcome function, especially in the case where there are many more agents than items. Specifically, Lemma \ref{lem:medianmisallocate} shows that misallocated effort could be as high as $\nplayer_a \cd \p{1 - \frac{1}{\nitem}}$, which can be much greater than the $2 \cd \nitem$ bound given in Lemma \ref{lem:meanmisallocate}.

\begin{restatable}{lemma}{medianmisallocate}\label{lem:medianmisallocate}
For median outcomes, worst-case misallocated effort is \emph{lower}-bounded by $0.25 \cd \nplayer$, given $\nplayer = \nplayer_a + \nplayer_b \geq 2 \cd \nitem$.  
\end{restatable}

This proof (deferred to Appendix \ref{app:proofs}) is constructive and creates an arrangement with high misallocated effort, based on the proof of Lemma \ref{lem:medianstableequal}. This lower bound involves an arrangement that is a (pure) Nash equilibrium under Definition \ref{def:nash} because no \emph{single} player can decrease its cost by changing which item it is occupying. However, if agents were allowed to coordinate, as in \emph{Colonel} Blotto, then such an arrangement may fail to be stable (e.g.  multiple agents of type one type could move together to another item and decrease their total cost). We view this result as illustrating how the decentralized nature of the Private Blotto game can cause arrangements with high misallocated cost to be stable, where in the centralized Colonel Blotto they would be unstable.

Overall, the results of this analysis imply that mean outcome functions, rather than median ones, give sharper guarantees that any stable arrangement that exists will result in agents roughly arranging themselves across items in proportion to their overall value.
 Specifically, our results have implications for the design of social media labeling tools. While these tools have outcome functions that are significantly more sophisticated than median or mean \cite{wojcik2022birdwatch}, in general our results suggest that more smooth outcome functions (similar to mean) as compared to sharper outcome functions (similar to median) would reduce the worst-case misallocated effort, at the expense of greater unpredictability as to when a stable allocation exists.

\section{More items: $\nitem>\nplayer$}\label{sec:feweragent}

In Section \ref{sec:moreagents}, we examined the case with more agents than items: $\nplayer\geq \nitem$. In this section, we will explore the other possibility, with fewer agents than items. This could model content annotation in settings where the number of prospective users is far fewer than the number of posts (a common occurrence). Our goal here is to model the set of stable arrangements, again comparing median and mean outcome functions. 
 
Because $\nplayer< \nitem$, some items will inevitably need to be left empty. Because of this, we will drop the lower bound in Lemma \ref{lem:nounlabeled} and allow the cost for leaving a item empty $\costun$ to be set arbitrarily. Since Lemma \ref{lem:nounlabeled} no longer applies, in this section we will see that agent biases $\{\bias_i\}$ are relevant for agents' strategies -- a departure from Section \ref{sec:moreagents}, where all of our results held independent of agent bias.

At a high level, Section \ref{sec:moreagents} showed a wide divergence in behavior between median and mean outcome functions. In this section, we will show that the setting of $\nplayer < \nitem$ gives much more similar results between the two outcome functions, though with some differences. The intuition is that both median and mean outcome functions behave identically for items with only 1 or 2 agents present. Because there are few agents compared to the number of items, most arrangements will have 1 or 2 agents per item, unless $\costun$ is very small or differences in biases is very large. Our main result in this section is Theorem \ref{thrm:whenstablefeweragents}, which exactly characterizes when a stable arrangement is guaranteed to exist and when it is possible that none exist for the $\nplayer < \nitem$ regime. In particular, note that mean and median outcome function have almost identical patterns of when stable arrangements exist, except for at $\nplayer=3$, and that cases where no stable arrangement exists are relatively common.  
\begin{theorem}\label{thrm:whenstablefeweragents}
Given $\nitem > \nplayer$, 
\begin{itemize}
    \item For $\nplayer = 2$, there is always a stable arrangement (for both median and mean outcome function). 
    \item For $\nplayer \in [4, \nitem)$, for both median and mean outcome function there are always parameters (player biases $\{\bias_a, \bias_b\}$ and unlabeled cost $\costun$) such that no stable arrangement exists. 
    \item For $\nplayer =3$, for mean outcome function, a stable arrangement always exists, but for median outcome function, there exists parameters such that no stable arrangement is possible.
\end{itemize}
\end{theorem}
We prove this theorem through a series of smaller lemmas. 

First, Lemma \ref{lem:twostable} considers the $\nplayer=2$ setting, showing that a stable arrangement always exists for both median and mean outcome.

\begin{restatable}{lemma}{twostable}\label{lem:twostable}
For $\nplayer = 2$, $\nitem \geq 2$ with either median or mean outcome functions, a stable arrangement always exists, regardless of the player biases $\{\bias_a, \bias_b\}$ and unlabeled cost $\costun$. 
\end{restatable}

Next, Lemmas \ref{lem:medianunstable} and \ref{lem:meanunstable} describe the complementary condition, showing when (for median and mean outcome functions respectively), stable arrangements fail to exist. The proofs (deferred to Appendix \ref{app:proofs}) are constructive and involve creating a set of agents with specific biases and cost for leaving an item unlabeled such that any possible arrangement has at least one agent wishing to label a different item. 

\begin{restatable}{lemma}{medianunstable}
\label{lem:medianunstable}
For any $\nplayer, \nitem$ with such that $2 < \nplayer<\nitem$, with median outcome, there exists biases $\{\beta_i\}$ and costs $\costun$ such that no NE exists. 
\end{restatable}

\begin{restatable}{lemma}{meanunstable}
\label{lem:meanunstable}
For any $\nplayer, \nitem$ such that $4 \leq \nplayer<\nitem$, with mean outcome, there exists parameters such that no NE exists. 
\end{restatable}

The mean outcome function case explored in Lemma \ref{lem:meanunstable} is not exactly analogous to the median case in Lemma \ref{lem:medianunstable}: there is a gap at $\nplayer<4$ players. Lemma \ref{lem:meanstablen3m4} completes this gap by showing that the gap in Lemma \ref{lem:meanunstable} is inevitable: any possible set of $\nplayer =3$ agents must have a stable arrangement, given mean outcomes.  

\begin{restatable}{lemma}{meanstablen}
\label{lem:meanstablen3m4}
For $\nplayer=3, \nitem \geq 4$ with mean outcome, there is always a stable arrangement.
\end{restatable}
Taken together, these results prove Theorem \ref{thrm:whenstablefeweragents}. At a high level, these results show that for almost all $\nplayer<\nitem$, it is possible that no stable arrangement of players over items exists. However, Lemma \ref{lem:cuNsmallNE} shows that a stable arrangement where every agent is labeling a separate is guaranteed to exist when the cost for leaving an item unlabeled is sufficiently high:

\begin{restatable}{lemma}{cuNsmallNE}
\label{lem:cuNsmallNE}
Given $\nplayer<\nitem$, an arrangement with all agents labeling different items is stable (for both median and mean outcome) so long as the cost for leaving an item unlabeled is sufficiently high: 
$$\costun \geq 0.5 \cd \abs{\bias_a - \beta_{b}} $$ 

\end{restatable}

Note that the condition in Lemma \ref{lem:cuNsmallNE} (when having exactly 1 agent per item is stable, given $\nitem>\nplayer$) is exactly the same as the condition in Lemma \ref{lem:nounlabeled} (when no item will ever be left unlabeled, given $\nitem \leq \nplayer$). This suggests that this level of $\costun$ could be viewed as a critical threshold which governs when certain types of arrangements are stable. 

Taken together, these results have implications for designers of real-life systems that behave like Private Blotto games. Specifically, Theorem \ref{thrm:whenstablefeweragents} suggests that for $\nplayer<\nitem$, the choice of outcome function (median or mean) is relatively unimportant. However, for almost all values of $\nplayer$, it is possible that no stable arrangement exists, unless the cost for leaving an item unlabeled is sufficiently high (Lemma \ref{lem:cuNsmallNE}). A designer of such a system could increase the cost of leaving an item unlabeled by proactively highlighting posts in need of notes (in the social media example) or giving more airtime to political issues that are under-debated (in the political competition example).

\section{Discussion}\label{sec:discussion}
In this paper, we proposed and analyzed the Private Blotto game, a multi-player game involving competition over items with different values. We focused on the impact of the outcome function on whether Nash stable arrangements exist. For the case with more agents than items, we showed that the choice of outcome function is critical. A median outcome function guarantees that the number of settings where no stable arrangement exists is finite and small compared to the number of items $\nitem$, but could involve high degrees of misallocated effort where agents are unevenly spread across items. By contrast, the mean outcome function does not guarantee stable arrangements exist, even for arbitrarily large number of players. However, when a stable arrangement exists, it is guaranteed to have low levels of misallocated effort. In Section \ref{sec:feweragent} we analyzed the case with more items than agents, showing that median and mean outcome functions behave much more similarly, and given sufficiently high cost for leaving an item unlabeled, all agents will spread evenly over items, minimizing misallocated effort. 

Throughout, we used motivating examples related to civic institutions and social welfare, such as detection of online misinformation. Our results have implications for how such tools should be developed, especially in the choice of the outcome function. Specifically, if there are many agents and guaranteeing a stable arrangement is important, median outcome (or a similar function) would likely be best, but if minimizing misallocated effort is more important, mean outcome (or a similar function) is likely better. For cases with fewer agents than items, increasing the cost of leaving items unlabeled is more important than the choice of outcome function. 
While in this work we addressed the primary question of Nash stability, there are multiple interesting extensions into the Private Blotto game. One natural extension would be to consider cases where some items are considered more important than others - for example, we may wish to have more human effort placed on discussing an important bill than on which shoes a celebrity wore to a gala. Another natural extension would be to consider more than two polarized types $A$ and $B$, potentially a continuum of biases $\{\beta_i\}$ reflecting a more nuanced set of opinions. We discuss how some of our results can be extended to both of these settings in Appendix \ref{app:unevenweight}. 

Finally, another natural extension could interpolate between the historic Colonel Blotto game and our novel Private Blotto game by allowing agents to coordinate with up to $\nplayer'>1$ other agents. For example, for $\nplayer'=2$, an arrangement would fail to be Nash stable if two agents, working together, could move and improve both of their utility. This modification would only make Nash stability harder to achieve, but could reflect scenarios with limited degrees of coordination. 

\ifarxiv
\subsubsection*{Acknowledgements}
This work was supported in part by a Simons Investigator Award, a Vannevar Bush Faculty Fellowship, MURI grant W911NF-19-0217, AFOSR grant FA9550-19-1-0183, ARO grant W911NF19-1-0057, a Simons Collaboration grant, a grant from the MacArthur Foundation, and NSF grant DGE-1650441. We are extremely grateful to Maria Antoniak, Sarah Dean, Jason Gaitonde, and the AI, Policy, and Practice working group at Cornell for invaluable discussions. 

\else 
\fi 

\clearpage 
\ifarxiv
\bibliographystyle{plainnat}
\bibliography{aaai25}
\else 
\bibliography{aaai25} 

\begin{thebibliography}{29}
\providecommand{\natexlab}[1]{#1}
\providecommand{\url}[1]{\texttt{#1}}
\expandafter\ifx\csname urlstyle\endcsname\relax
  \providecommand{\doi}[1]{doi: #1}\else
  \providecommand{\doi}{doi: \begingroup \urlstyle{rm}\Url}\fi

\bibitem[Adam et~al.(2021)Adam, Hor{\v{c}}{\'\i}k, Kasl, and Kroupa]{adam2021double}
Luk{\'a}{\v{s}} Adam, Rostislav Hor{\v{c}}{\'\i}k, Tom{\'a}{\v{s}} Kasl, and Tom{\'a}{\v{s}} Kroupa.
\newblock Double oracle algorithm for computing equilibria in continuous games.
\newblock In \emph{Proceedings of the AAAI Conference on Artificial Intelligence}, volume~35, pages 5070--5077, 2021.

\bibitem[Ahmadinejad et~al.(2016)Ahmadinejad, Dehghani, Hajiaghayi, Lucier, Mahini, and Seddighin]{ahmadinejad2019duels}
AmirMahdi Ahmadinejad, Sina Dehghani, MohammadTaghi Hajiaghayi, Brendan Lucier, Hamid Mahini, and Saeed Seddighin.
\newblock From duels to battlefields: Computing equilibria of blotto and other games.
\newblock \emph{Proceedings of the AAAI Conference on Artificial Intelligence}, 2016.

\bibitem[Allen et~al.(2022)Allen, Martel, and Rand]{birdsdontfactcheck}
Jennifer Allen, Cameron Martel, and David~G Rand.
\newblock Birds of a feather don’t fact-check each other: Partisanship and the evaluation of news in twitter’s birdwatch crowdsourced fact-checking program.
\newblock New York, NY, USA, 2022. Association for Computing Machinery.
\newblock ISBN 9781450391573.
\newblock \doi{10.1145/3491102.3502040}.
\newblock URL \url{https://doi.org/10.1145/3491102.3502040}.

\bibitem[Anbarc{\i} et~al.(2020)Anbarc{\i}, Cingiz, and Ismail]{anbarci2020proportional}
Nejat Anbarc{\i}, Kutay Cingiz, and Mehmet~S Ismail.
\newblock Proportional resource allocation in dynamic n-player blotto games.
\newblock \emph{arXiv preprint arXiv:2010.05087}, 2020.

\bibitem[Anshelevich et~al.(2008)Anshelevich, Dasgupta, Kleinberg, Tardos, Wexler, and Roughgarden]{anshelevich2008price}
Elliot Anshelevich, Anirban Dasgupta, Jon Kleinberg, {\'E}va Tardos, Tom Wexler, and Tim Roughgarden.
\newblock The price of stability for network design with fair cost allocation.
\newblock \emph{SIAM Journal on Computing}, 38\penalty0 (4):\penalty0 1602--1623, 2008.

\bibitem[Behnezhad et~al.(2017)Behnezhad, Dehghani, Derakhshan, Hajiaghayi, and Seddighin]{Behnezhad_Dehghani_Derakhshan_Hajiaghayi_Seddighin_2017}
Soheil Behnezhad, Sina Dehghani, Mahsa Derakhshan, MohammadTaghi Hajiaghayi, and Saeed Seddighin.
\newblock Faster and simpler algorithm for optimal strategies of blotto game.
\newblock \emph{Proceedings of the AAAI Conference on Artificial Intelligence}, 31\penalty0 (1), Feb. 2017.
\newblock \doi{10.1609/aaai.v31i1.10620}.
\newblock URL \url{https://ojs.aaai.org/index.php/AAAI/article/view/10620}.

\bibitem[Boix-Adser{\`a} et~al.(2020)Boix-Adser{\`a}, Edelman, and Jayanti]{boix2020multiplayer}
Enric Boix-Adser{\`a}, Benjamin~L Edelman, and Siddhartha Jayanti.
\newblock The multiplayer colonel blotto game.
\newblock In \emph{Proceedings of the 21st ACM Conference on Economics and Computation}, pages 47--48, 2020.

\bibitem[Borel(1921)]{borel1921theorie}
Emile Borel.
\newblock La th{\'e}orie du jeu et les {\'e}quations int{\'e}gralesa noyau sym{\'e}trique.
\newblock \emph{Comptes rendus de l’Acad{\'e}mie des Sciences}, 173\penalty0 (1304-1308):\penalty0 58, 1921.

\bibitem[Gil~de Zúñiga et~al.(2012)Gil~de Zúñiga, Jung, and Valenzuela]{10.1111/j.1083-6101.2012.01574.x}
Homero Gil~de Zúñiga, Nakwon Jung, and Sebastián Valenzuela.
\newblock {Social Media Use for News and Individuals' Social Capital, Civic Engagement and Political Participation}.
\newblock \emph{Journal of Computer-Mediated Communication}, 17\penalty0 (3):\penalty0 319--336, 04 2012.
\newblock ISSN 1083-6101.
\newblock \doi{10.1111/j.1083-6101.2012.01574.x}.
\newblock URL \url{https://doi.org/10.1111/j.1083-6101.2012.01574.x}.

\bibitem[Golman and Page(2009)]{golman2009general}
Russell Golman and Scott~E Page.
\newblock General blotto: games of allocative strategic mismatch.
\newblock \emph{Public Choice}, 138\penalty0 (3):\penalty0 279--299, 2009.

\bibitem[Goyal and Vigier(2014)]{attackdefense}
Sanjeev Goyal and Adrien Vigier.
\newblock {Attack, Defence, and Contagion in Networks}.
\newblock \emph{The Review of Economic Studies}, 81\penalty0 (4):\penalty0 1518--1542, 07 2014.
\newblock ISSN 0034-6527.
\newblock \doi{10.1093/restud/rdu013}.
\newblock URL \url{https://doi.org/10.1093/restud/rdu013}.

\bibitem[Hart(2008)]{hart2008discrete}
Sergiu Hart.
\newblock Discrete colonel blotto and general lotto games.
\newblock \emph{International Journal of Game Theory}, 36\penalty0 (3):\penalty0 441--460, 2008.

\bibitem[Hettiachchi et~al.(2022)Hettiachchi, Kostakos, and Goncalves]{hettiachchi2022survey}
Danula Hettiachchi, Vassilis Kostakos, and Jorge Goncalves.
\newblock A survey on task assignment in crowdsourcing.
\newblock \emph{ACM Computing Surveys (CSUR)}, 55\penalty0 (3):\penalty0 1--35, 2022.

\bibitem[Koutsoupias and Papadimitriou(2009)]{koutsoupias2009worst}
Elias Koutsoupias and Christos Papadimitriou.
\newblock Worst-case equilibria.
\newblock \emph{Computer science review}, 3\penalty0 (2):\penalty0 65--69, 2009.

\bibitem[Kovenock and Roberson(2012)]{kovenock2012coalitional}
Dan Kovenock and Brian Roberson.
\newblock Coalitional colonel blotto games with application to the economics of alliances.
\newblock \emph{Journal of Public Economic Theory}, 14\penalty0 (4):\penalty0 653--676, 2012.

\bibitem[math.stackexchange.com(2015)]{stackoverflow}
math.stackexchange.com.
\newblock Is the optimal solution of a strictly convex function over $\mathbb{Z}^d$ a rounded version of its optimal solution over $\mathbb{R}^d$, 2015.
\newblock URL \url{https://math.stackexchange.com/questions/1213609/is-the-optimal-solution-of-a-strictly-convex-function-over-mathbbzd-a-roun}.

\bibitem[Mazur(2017)]{mazur2017partial}
Kostyantyn Mazur.
\newblock A partial solution to continuous blotto.
\newblock \emph{arXiv preprint arXiv:1706.08479}, 2017.

\bibitem[Merolla et~al.(2005)Merolla, Munger, and Tofias]{merolla2005play}
Jennifer Merolla, Michael Munger, and Michael Tofias.
\newblock In play: A commentary on strategies in the 2004 us presidential election.
\newblock \emph{Public Choice}, 123:\penalty0 19--37, 2005.

\bibitem[Nielsen and Schrøder(2014)]{doi:10.1080/21670811.2013.872420}
Rasmus~Kleis Nielsen and Kim~Christian Schrøder.
\newblock The relative importance of social media for accessing, finding, and engaging with news.
\newblock \emph{Digital Journalism}, 2\penalty0 (4):\penalty0 472--489, 2014.
\newblock \doi{10.1080/21670811.2013.872420}.
\newblock URL \url{https://doi.org/10.1080/21670811.2013.872420}.

\bibitem[Osorio(2013)]{OSORIO2013164}
Antonio Osorio.
\newblock The lottery blotto game.
\newblock \emph{Economics Letters}, 120\penalty0 (2):\penalty0 164--166, 2013.
\newblock ISSN 0165-1765.
\newblock \doi{https://doi.org/10.1016/j.econlet.2013.04.012}.
\newblock URL \url{https://www.sciencedirect.com/science/article/pii/S0165176513001833}.

\bibitem[Pr{\"o}llochs(2022)]{prollochs2022community}
Nicolas Pr{\"o}llochs.
\newblock Community-based fact-checking on twitter’s birdwatch platform.
\newblock In \emph{Proceedings of the International AAAI Conference on Web and Social Media}, volume~16, pages 794--805, 2022.

\bibitem[Saeed et~al.(2022)Saeed, Traub, Nicolas, Demartini, and Papotti]{saeed2022crowdsourced}
Mohammed Saeed, Nicolas Traub, Maelle Nicolas, Gianluca Demartini, and Paolo Papotti.
\newblock Crowdsourced fact-checking at twitter: How does the crowd compare with experts?
\newblock In \emph{Proceedings of the 31st ACM International Conference on Information \& Knowledge Management}, pages 1736--1746, 2022.

\bibitem[Schwartz et~al.(2014)Schwartz, Loiseau, and Sastry]{schwartz2014heterogeneous}
Galina Schwartz, Patrick Loiseau, and Shankar~S Sastry.
\newblock The heterogeneous colonel blotto game.
\newblock In \emph{2014 7th international conference on NETwork Games, COntrol and OPtimization (NetGCoop)}, pages 232--238. IEEE, 2014.

\bibitem[Skaperdas(1996)]{skaperdas1996contest}
Stergios Skaperdas.
\newblock Contest success functions.
\newblock \emph{Economic theory}, 7:\penalty0 283--290, 1996.

\bibitem[Tullock(2001)]{Tullock2001}
Gordon Tullock.
\newblock \emph{Efficient Rent Seeking}, pages 3--16.
\newblock Springer US, Boston, MA, 2001.
\newblock ISBN 978-1-4757-5055-3.
\newblock \doi{10.1007/978-1-4757-5055-3_2}.
\newblock URL \url{https://doi.org/10.1007/978-1-4757-5055-3_2}.

\bibitem[Wojcik et~al.(2022)Wojcik, Hilgard, Judd, Mocanu, Ragain, Hunzaker, Coleman, and Baxter]{wojcik2022birdwatch}
Stefan Wojcik, Sophie Hilgard, Nick Judd, Delia Mocanu, Stephen Ragain, MB~Hunzaker, Keith Coleman, and Jay Baxter.
\newblock Birdwatch: Crowd wisdom and bridging algorithms can inform understanding and reduce the spread of misinformation.
\newblock \emph{arXiv preprint arXiv:2210.15723}, 2022.

\bibitem[Yasseri and Menczer(2021)]{yasseri2021can}
Taha Yasseri and Filippo Menczer.
\newblock Can the wikipedia moderation model rescue the social marketplace of ideas?
\newblock 2021.

\bibitem[Zhang et~al.(2017)Zhang, Sheng, Li, Wu, and Wu]{zhang2017consensus}
Jing Zhang, Victor~S Sheng, Qianmu Li, Jian Wu, and Xindong Wu.
\newblock Consensus algorithms for biased labeling in crowdsourcing.
\newblock \emph{Information Sciences}, 382:\penalty0 254--273, 2017.

\bibitem[Zhang and van~der Schaar(2012)]{6195597}
Yu~Zhang and Mihaela van~der Schaar.
\newblock Reputation-based incentive protocols in crowdsourcing applications.
\newblock In \emph{2012 Proceedings IEEE INFOCOM}, pages 2140--2148, 2012.
\newblock \doi{10.1109/INFCOM.2012.6195597}.

\end{thebibliography}
\fi 

\clearpage 
\ifarxiv

\else 

This paper:

\begin{itemize}
\item Includes a conceptual outline and/or pseudocode description of AI methods introduced:  yes
\item Clearly delineates statements that are opinions, hypothesis, and speculation from objective facts and results: yes
\item Provides well marked pedagogical references for less-familiare readers to gain background necessary to replicate the paper: yes
\item Does this paper make theoretical contributions? yes
\end{itemize}

If yes, please complete the list below.

\begin{itemize}
\item All assumptions and restrictions are stated clearly and formally: yes
\item All novel claims are stated formally (e.g., in theorem statements): yes
\item Proofs of all novel claims are included: yes
\item Proof sketches or intuitions are given for complex and/or novel results: partial
\item Appropriate citations to theoretical tools used are given. yes
\item All theoretical claims are demonstrated empirically to hold: partial. 
\item All experimental code used to eliminate or disprove claims is included: Included in supplementary. 
\item Does this paper rely on one or more datasets? NO
\end{itemize}

If yes [experiments/computation], please complete the list below.

\begin{itemize}
\item Any code required for pre-processing data is included in the appendix: N/A.  
\item All source code required for conducting and analyzing the experiments is included in a code appendix: In supplementary material. 
\item All source code required for conducting and analyzing the experiments will be made publicly available upon publication of the paper with a license that allows free usage for research purposes. Yes
\item All source code implementing new methods have comments detailing the implementation, with references to the paper where each step comes from: yes
\item If an algorithm depends on randomness, then the method used for setting seeds is described in a way sufficient to allow replication of results. N/A
\item This paper specifies the computing infrastructure used for running experiments (hardware and software), including GPU/CPU models; amount of memory; operating system; names and versions of relevant software libraries and frameworks. N/A, code runs in 10 minutes on a laptop. 
\item This paper formally describes evaluation metrics used and explains the motivation for choosing these metrics. Yes
\item This paper states the number of algorithm runs used to compute each reported result. Yes (1). 
\item Analysis of experiments goes beyond single-dimensional summaries of performance (e.g., average; median) to include measures of variation, confidence, or other distributional information. N/A
\item The significance of any improvement or decrease in performance is judged using appropriate statistical tests (e.g., Wilcoxon signed-rank). N/A (no randomness). 
\item This paper lists all final (hyper-)parameters used for each model/algorithm in the paper’s experiments. N/A
\item This paper states the number and range of values tried per (hyper-) parameter during development of the paper, along with the criterion used for selecting the final parameter setting. N/A
\end{itemize}
\fi 
\appendix
\clearpage 

\section{Further related works}\label{app:related}

\subsection{Colonel Blotto}\label{app:blottorelated}
Colonel Blotto games (first proposed in \cite{borel1921theorie}) is a game theoretic model where two different players, $A$ and $B$, compete to allocate effort across multiple \enquote{battlefields} or \enquote{items}, which may vary in how much each player values them. Typically, an player wins an item if they exert more effort there, and one main question of interest is when Nash equilibria of this system exist. The literature on Colonel Blotto games is extremely broad, so will focus on a few of the most relevant papers. Recently, \cite{ahmadinejad2019duels} included a polynomial time algorithm for computing equilibria of the standard Colonel Blotto game, as well as related zero-sum games, a line of work that has been further extended \cite{Behnezhad_Dehghani_Derakhshan_Hajiaghayi_Seddighin_2017, adam2021double}. 

First, we will highlight some of the most commonly-studied variants of Colonel Blotto games. 
\cite{golman2009general} proposes the \enquote{General Blotto} game, which generalizes Colonel Blotto to permit multiple player types which have smooth utility functions over items and over combinations of items. \cite{hart2008discrete} proposes the \enquote{General Lotto} game where each player selects a probabilistic distribution of effort over each items and gets utility given by the probability that a randomly drawn level of their effort beats their opponents' random draw. Separately, \cite{OSORIO2013164} proposes the \enquote{Lottery Blotto} game, where players allocate effort deterministically, but the player that allocates greater effort to an item doesn't win deterministically, but rather probabalistically. This formulation is related to Tullock Contests Success functions (originally proposed in \cite{Tullock2001}, also studied in \cite{OSORIO2013164}, \cite{skaperdas1996contest}) where two players are competing over contests where they win probabalistically related to their effort (similar to our mean outcome function). \cite{attackdefense} similarly studies contest functions where items are connected in a network and an asymmetric \enquote{attacker} and \enquote{defender} are allocating resources across these items. 

Next, we will highlight the papers that are closest to ours. \cite{schwartz2014heterogeneous} gives Nash stability results for the Colonel Blotto game where players vary in their strength (amount of resources) and items vary in their value, so long as there are at least three items with each value. \cite{kovenock2012coalitional} studies a limited form of coalitions where exactly two players $A$ and $B$ may form an alliance before playing a common opponent $C$. \cite{boix2020multiplayer} proposes and studies the \enquote{multi-player Colonel Blotto game}, which extends the classical Colonel Blotto structure to more than 2 players. \cite{anbarci2020proportional} studies a variant of Colonel Blotto with more than two competing forces, but where items are presented sequentially, rather than simultaneously. \cite{mazur2017partial} studies Nash equilibria of Colonel Blotto games with exactly 2 items, but where outcome functions are constrained to be a polynomial function of the difference of each type's allocation across the items. 

In general, Colonel Blotto games and their variants differ from ours in a few ways. First, rather than assuming the player can control multiple agents, we assume each agent acts independently (a private citizen as opposed to a soldier). Because any of these agents could \enquote{win}, this dramatically increases the number of potential outcomes. However, in our second main difference, we assume that agents have some degree of similarity in their goals: agent $A$ may be more closely aligned with $B$ than $C$, for example. Finally, we study a more general class of settings than is typically studied in Colonel Blotto, allowing for arbitrary numbers and valuations of items, as well as more general notions of winning (all or nothing, as well as a more smooth fractional utility). 

\subsection{Crowdsourcing}\label{app:crowdrelated}
The area of crowdsourcing has been studied experimentally and theoretically in a wide range of papers. Again, we will focus on summarizing those that are most closely related to ours. Some papers, such as \cite{hettiachchi2022survey}, \cite{zhang2017consensus} study how to assign different crowd workers to multiple tasks in order to maximize the expected accuracy of labels. 

A more nascent branch of crowdsourcing considers the case where crowd workers may have agency over which items they choose to label. \cite{6195597} studies reputation-based mechanisms to incentivize crowd workers to exert effort on items that they are assigned. Our model is especially relevant to fact-checking on social media sites that allow voluntary labels by (potentially biased) users, such as on Facebook, Wikipedia and Twitter \cite{yasseri2021can}, especially Twitter's BirdWatch tool \cite{wojcik2022birdwatch}. \cite{birdsdontfactcheck}\cite{prollochs2022community} study how partisan affiliation, among other facts, affects how users on Twitter choose which tweets to fact check. \cite{saeed2022crowdsourced} compares the accuracy of labels produced by voluntary, biased crowdsource workers to expert labels. 

Our paper differs from most crowdsourcing papers in how it allows crowdsource workers to act as voluntary, potentially biased agents with some agency over which items they choose to label. Our work also differs stylistically in that we focus primarily on Nash Equilibria of such systems, an area that has typically not been explored previously. 

\section{Arbitrary weights over items, arbitrary numbers of agent types (biases)}\label{app:unevenweight}
In the main body of the text, we proposed and studied Private Blotto, a variant of the celebrated Colonel Blotto game when agents move in a disaggregated fashion, without a central Colonel. We saw that questions of Nash stability were already rich and complex even when all agents belong to one of two types and items have equal weight, just as Colonel Blotto games already exhibit rich structure with just two colonels and when all items (fronts) have equal importance. 

However, there are natural extensions of the Private Blotto setting that are especially relevant in the motivating examples of social media platform annotation and political issue competition. In this section, we propose a few extensions of the Private Blotto game and discuss how our results would generalize. Specifically, we will discuss two changes: 
\begin{enumerate}
\item First, we allow different items can have different levels of importance (\emph{weights}, $\weight_i \in [\nitem]$). This could represent settings when certain items are commonly agreed to have greater importance, and thus winning them becomes more sought-after. Given this, the cost for agents becomes a weighted sum over items: 
$$\sum_{i \in [\nitem], \abs{S_i}>0}\weight_i \cd \abs{f(S_i) - \bias_{\type}} + \sum_{i \in [\nitem], \abs{S_i} =0} \weight_i \cd \costun$$
\item Second, we allow agents to come in arbitrarily many different biases $\{\bias_i\}$, rather than exactly two groups. This could model a continuum of opinions, for example, or alliances among multiple differing subgroups. 
\end{enumerate}

Next, we discuss how some of our results could generalize, given these extensions of the Private Blotto model. We follow the organization of the main paper, and divide our generalizations into two broad cases: when there are more agents than items (so that $\nplayer \geq \nitem$, in Appendix \ref{app:moreagents}) and when there are more items than agents (so that $\nitem > \nplayer$, in Appendix \ref{app:moreitem}). All proofs for this section are given in Appendix \ref{app:proofsunevenweight}.

\subsection{More agents (extension of Section \ref{sec:moreagents})}\label{app:moreagents}

First, Lemma \ref{lem:nounlabeledw} generalizes the condition for which no item will be left unlabeled. 
\begin{restatable}{lemma}{nounlabeledw}
\label{lem:nounlabeledw}
If there are more agents than items ($\nplayer \geq \nitem$) and the cost for leaving a item empty is sufficiently high, then no item will be left empty, regardless of if median or mean outcome function is used. Specifically, this occurs when:  
$$\costun\geq \frac{\max_{i \in \nitem}\weight_i}{\min_{j \in \nitem}\weight_j} \cd \max_{k, \ell \in[\type]} \abs{\bias_k - \bias_{\ell}} \cd \frac{1}{2}$$ 
Moreover, if there are exactly two types of agents with biases $\bias_a, \bias_b$, then for both median and mean outcome functions, agent strategy becomes independent of biases $\bias_a, \bias_b$ and relies solely on the number of agents of each type on each item, $\{a_i, b_i\}, \ i \in [\nitem]$. 
\end{restatable}
Note that if there are more than two types of player biases, then agent strategy is no longer independent of player biases. Thus, extending to the setting with $\ntype>2$ would likely substantially change our results. We defer this to future work and focus on describing how importance weights on items would impact our results. 
\subsubsection{Median outcome function}
In this section, many of our results will directly translate to cases with weighted items, though a few results will change in this more general setting. 

First, we note that Lemmas \ref{lem:medianneverstable} and \ref{lem:median2mstable} already apply equally well to the weighted case: they give conditions for when an arrangement is stable when no player can unilaterally reduce their cost (e.g. go from losing a front to tying, or from tying to winning): doing so is orthogonal to the weights on individual fronts.

Lemma \ref{lem:exmednotstab}, below, demonstrates why the equal weights condition in Lemma \ref{lem:medianstableequal} is necessary. If even the weight of one item is even slightly higher than the weight in another, there exist conditions where no stable arrangement exists. 

\begin{restatable}{lemma}{exmednotstab}
\label{lem:exmednotstab}
Set $\nitem = 3, \nplayer_a = \nplayer_b = 3$, with $w_1 = w_2 + \epsilon$ and $w_2 = w_3 = w_4$, and and cost satisfying Lemma \ref{lem:nounlabeled}.  Then, the arrangement proposed by Lemma \ref{lem:medianstableequal} is not stable, and moreover, there is no possible stable arrangement. 
\end{restatable}

Finally, Lemma \ref{lem:medianstablerelaxequal} extends Lemma \ref{lem:medianstableequal} by relaxing the requirement that the item weights be exactly equal. Instead, this proof shows that it is sufficient to have the two items with highest weight be equal, while no other two items differ in weight by more than a factor of 2. Taken together, these lemmas prove Theorem \ref{thrm:medianmoreagents} and characterize the stability of the Private Blotto game with median outcome function and more agents than items. 

\begin{restatable}{lemma}{medianstablerelaxequal}
\label{lem:medianstablerelaxequal}[Extension of Lemma \ref{lem:medianstableequal}]
Any other number of agents ($\nplayer_a, \nplayer_b$) with $\nplayer_a + \nplayer_b \leq 2 \cd \nitem$ (besides those in the median-critical region) always has a stable arrangement, given cost satisfying Lemma \ref{lem:nounlabeled} and weights in descending order satisfying 
$\weight_0 = \weight_1 \text{ and }  \weight_i \leq 2 \cd \weight_j \ \forall i, j \in [\nitem]$. 
\end{restatable}

\subsubsection{Mean outcome function}
In Section \ref{sec:moreagentsmean} we included Figure \ref{fig:stable}, which illustrated when no stable arrangement existed for $\nitem=2$ items with equal weight. This figure showed that there persistent (in)stability, even for large numbers of agents $\nplayer_a, \nplayer_b$. Theorem \ref{thrm:meanbigunstable} formalized this intuition by showing that for all $\nplayer_a, \nplayer_b$, there existed a pair $\nplayer_a'\geq \nplayer_a, \nplayer_b' \geq \nplayer_b$ such that no stable arrangement exists for $\nplayer_a', \nplayer_b'$ players. 

In this section, we will explore how this story changes when the items are allowed to have different weights. In particular, we will show that the core result changes: given sufficiently unequal weights, there always exists a stable arrangement. First, Figure \ref{fig:stable_diffw} illustrates this numerically, showing that as the weights on two different items become more unequal, eventually every $(\nplayer_a, \nplayer_b)$ combination has a stable arrangement. Next, Theorem \ref{thrm:meanbigunstablew} formally proves this: for every $(\nplayer_a, \nplayer_b)$ pair, there exists a set of weights $\weight_1,\weight_2$ such that there is always a stable arrangement. 

\begin{figure}
    \centering
    \includegraphics[width=3in]{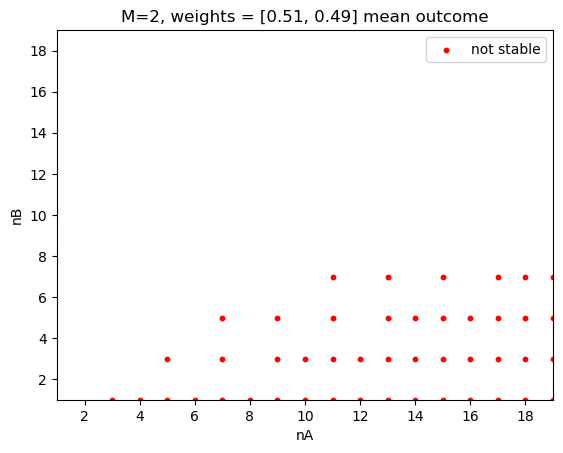}
    \includegraphics[width=3in]{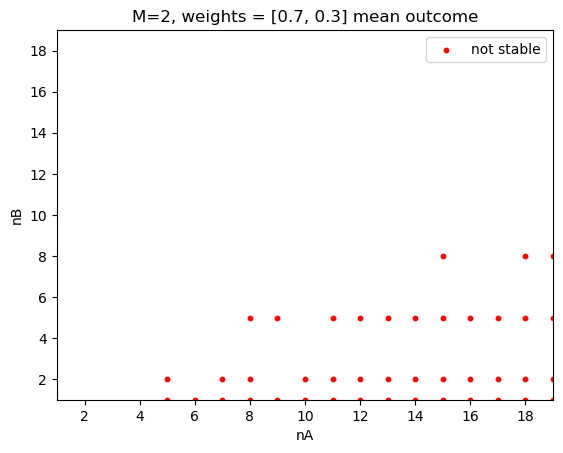}
    \includegraphics[width=3in]{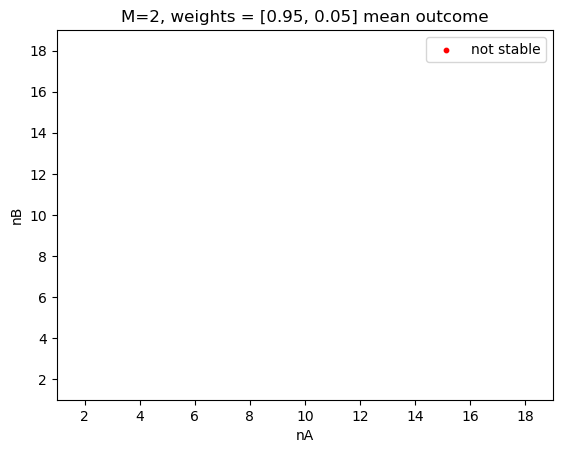}
    \caption{Version of Figure \ref{fig:stable} illustrating when stable arrangements of agents onto items exist for mean outcome ($\nitem=2$ items), but with differing weights on the $\nitem=2$ items. For clarity, only displayed for $\nplayer_a \geq \nplayer_b$. }
    \label{fig:stable_diffw}
\end{figure}

\begin{restatable}{theorem}{meanbigunstablew}\label{thrm:meanbigunstablew}
For all $\nplayer_a\geq \nplayer_b\geq 1$, there always exists some weights $\weight_1, \weight_2$ such that there is always a stable arrangement of players onto items (given cost $\costun$ satisfying the conditions of Lemma \ref{lem:nounlabeledw}).
\end{restatable}

Theorem \ref{thrm:meanbowlw} generalizes Theorem \ref{thrm:meanbowl}, showing that if players could be allocated fractionally, the proportional arrangement would always be stable. 

\begin{restatable}{theorem}{meanbowlw}
\label{thrm:meanbowlw}
For $\nitem$ items with two types of agents, $A$ and $B$ with mean outcome and $\costun$ satisfying the conditions of Lemma \ref{lem:nounlabeled}, if players are allowed to be allocated fractionally over items, then the stable arrangement is always given by $a_i = \weight_i \cd \nplayer_a, b_i = \weight_i \cd \nplayer_b$. 
\end{restatable}

For equal weights, Theorem \ref{thrm:meanclose} proved that in the integer-valued Private Blotto game, every stable arrangement is still \enquote{close} to evenly splitting agents over items. The proof of Theorem \ref{thrm:meanclose} strongly relies on equal weights, and we suspect that generalizing it to unequal weights would require a substantially different proof technique. However, empirically we observe that it holds even for unequal weights: every stable arrangement we have found computationally is \enquote{close} to proportional. Therefore, we hypothesize that this property holds more generally and discuss the implications of this hypothesis. 
\begin{hypothesis}\label{hyp:close}
Any stable arrangement in the integer Private Blotto game with mean outcome function and unequal weights must be \enquote{close} to proportional: $\abs{a_i - \weight_i \cd \nplayer_a} \leq 1, \abs{b_i - \weight_i \cd \nplayer_b} \leq 1$ for $i \in [\nitem]$. 
\end{hypothesis}

\subsubsection{Misallocated effort}

Definition \ref{def:misallocate_w} generalizes the definition of misallocated effort in Definition \ref{def:misallocate} to reflect the normative belief that items with greater weight (importance) should recieve more effort. 
\begin{definition}\label{def:misallocate_w}
Given an arrangement of agents onto items, we say it has \enquote{misallocated effort} given by the amount of agents that is above or below allocation equal to the weights. That is, misallocated effort is given by: 
$$\sum_{i \in [\nitem]}\abs{\weight_i \cd \nplayer_a - a_i} + \abs{\weight_i \cd \nplayer_b - b_i}$$
\end{definition}

For mean outcome function, if Hypothesis \ref{hyp:close} holds, then we can automatically see that misallocated effort is upper bounded by $2 \cd \nitem$ (as Lemma \ref{lem:meanmisallocate} shows for equal weights). 

For median outcome function, Lemma \ref{lem:medianmisallocate} gave a \emph{lower bound} on misallocated effort, which was obtained for the setting with equal weights, but which also gives a lower bound on misallocated effort with arbitrary weights.  

\subsection{More items (extension of Section \ref{sec:feweragent})}\label{app:moreitem}
In this setting, we explore generalizations both in weights over items, as well as number of player types. At a high level, the results in this section are largely very similar to the results in Section \ref{sec:feweragent}.

First, Lemma \ref{lem:twostablew} directly generalizes Lemma \ref{lem:twostable} in the main text, showing that whenever there are exactly 2 agents, a stable arrangement always exists (even in the weighted case). 

\begin{restatable}{lemma}{twostablew}\label{lem:twostablew}
For $\nplayer = 2$, $\nitem \geq 2$ and either median or mean outcome functions, a stable arrangement always exists, regardless of the player biases $\{\bias_i\}$ unlabeled cost $\costun$, and weight $\{\weight_i\}$. 
\end{restatable}

Lemmas \ref{lem:medianunstable} and \ref{lem:meanunstable} in the main text already apply to settings with multiple biases and uneven weights: they show that there can be settings with no stable arrangement with equal weights, which directly implies that there (in the arbitrary weights setting), it is possible that are settings with no stable arrangements. 

Lemma \ref{lem:meanstablen3m4} in the main text shows that with exactly 3 players over at least 4 items, equally weighted (which come in exactly two types, $A$ and $B$), there always exists a stable arrangement with mean outcome function. Lemma \ref{lem:meanstablen3m4_w}, below, generalizes this to settings where the 3 players are allowed to each have distinct biases. This generalization is non-trivial because it involves reasoning over substantially more complex combinations of players and incentives. The most general version of Lemma \ref{lem:meanstablen3m4_w} would involve arbitrary weights over items: however, if weights can be set arbitrarily, it can become substantially easier (and thus less interesting) to ensure a stable arrangement (e.g. by ensuring almost all players wish to be on the same item, as in Theorem \ref{thrm:meanbigunstablew}).  

\begin{restatable}{lemma}{meanstablenw}
\label{lem:meanstablen3m4_w}
For $\nplayer=3, \nitem \geq 4$ with mean outcome, there is always a stable arrangement, even when all three players can have distinct biases $\bias_a, \bias_b, \bias_c$. 
\end{restatable}

Finally, Lemma \ref{lem:cuNsmallNEw} directly generalizes Lemma \ref{lem:cuNsmallNE} for arbitrary player types and item weights. 

\begin{restatable}{lemma}{cuNsmallNEw}
\label{lem:cuNsmallNEw}
Given $\nplayer<\nitem$, an arrangement with all agents labeling different items is stable (for both median and mean outcome) so long as the cost for leaving an item unlabeled is sufficiently high: 

$$\abs{\beta_1 - \beta_{i*}} \leq  2\cd \frac{\weight{_\nplayer}}{\weight_1} \cd  \costun$$
\end{restatable}

\section{Proofs for Main Body}\label{app:proofs}

\begin{algorithm}
\caption{Algorithm for stable arrangement, not in median-critical region, given $\nplayer_a + \nplayer_b \geq 2 \cd \nitem+1$ or $\nplayer_a + \nplayer_b = 2 \cd \nitem $ with $\nplayer_a, \nplayer_b$ even. }
\label{alg:allocatemore}
Set $a_i = b_i = 0 \ \forall i \in [\nitem]$\\
\While{$\sum_{i \in [\nitem]}b_i < \nplayer_b$}{
\For{$j \in [\nitem-1]$}{
\lIf(){$j = \nitem-1$}{$b_j = \nplayer_b - \sum_{i  \in [\nitem]} b_i, \ j^* = j$}
\lElseIf(If at least 3 players left, allocate 2 per item){$\nplayer_b - \sum_{i  \in [\nitem]}b_i\geq 3$}{$b_j=2$}
\lElse(If there are exactly 3 left, allocate them all and stop){$b_j=3$, $j^*=j$}
}
}
\While{$\sum_{i \in [\nitem]}a_i < \nplayer_a$}{
\For(Start with the next item after type $B$ is allocated){$j \in [j^*+1, \nitem]$}{
\lIf(If on last item, allocate all remaining players){$j = \nitem$}{$a_i = \nplayer_a - \sum_{i  \in [\nitem]}a_i$}
\lElseIf(If at least 3 players left, allocate 2 per item){$\nplayer_a - \sum_{i  \in [\nitem]}a_i\geq 3$}{$a_j=2$}
\lElse(If there are exactly 3 left, allocate them all and stop){$a_j=3$}

}

}
\end{algorithm}

\begin{algorithm}
\caption{Algorithm for stable arrangement, not in median-critical region, given $\nplayer_a + \nplayer_b \leq 2 \cd \nitem$}
\label{alg:allocatefewer}
Set $a_i = b_i = 0 \ \forall i \in [\nitem]$\\
\If(Put exactly 1 player per item){$\nplayer_a + \nplayer_b = \nitem$}{
$a_i = 1 \ \forall i \in [\nplayer_a], b_i = 1 \ \forall i \in [\nplayer_a + 1, \nitem]$
}\uElseIf{$\nplayer_a+ \nplayer_b - \nitem$ is odd}{
$x = 0.5 \cd (\nplayer_a+ \nplayer_b - \nitem)+1$\\
$a_1 = b_1 = x$ Exact tie on first item\\
$a_i = 1 \ \forall i \in [2, \nplayer_a - x + 1]$ Put remaining $\nplayer_a -x$, $\nplayer_b-x$ players each on a single item. \\
$b_i = 1 \ \forall i \in [\nplayer_a - x + 2, \nitem]$\\

}\uElse{
$x = 0.5 \cd (\nplayer_a+ \nplayer_b - \nitem)$\\
$a_1 = b_1 = x$ Exact tie on first and second item. \\
$a_2 = b_2 = 1$\\
$a_i = 1 \ \forall i \in [3, \nplayer_a - x + 2]$ Put remaining $\nplayer_a -x -1, \nplayer_b -x -1$ agents each on a single item. \\
$b_i = 1 \ \forall i \in [\nplayer_a - x + 3, \nitem]$ \\

}
\end{algorithm}

\nounlabeled*
\begin{proof}
In other words, we want to ensure that for any player $\type$, the cost of competing in any item $i$ (leaving 
any other item $j$ empty) is higher than the cost of leaving item $i$ to competing in item $j$ alone:  

$$ \abs{f(S_i) - \bias_{\type}} + \costun \geq \abs{f(S_i \setminus \bias_{\type})}+ 0$$

First, we will analyze the case with mean outcome function. For an agent of type $A$, the cost it experiences from an item with $a$ agents of type $A$ and $b$ agents of type $B$ is given by: 
$$\abs{\frac{a \cd \bias_a + b \cd \bias_b}{a + b} - \bias_a} = \abs{\frac{a \cd \bias_a + b \cd \bias_b - (a+b) \cd \bias_a}{a+b}}$$
$$= \frac{b}{a+b} \cd \abs{\bias_a - \bias_b} $$
By identical reasoning, the cost to an agent of type $B$ is: 
$$ \frac{a}{a+b} \cd \abs{\bias_a - \bias_b} $$
Note that this construction immediately tells us that agent strategy must be independent of biases. For every item with $a+b>0$, an agent's cost is solely a function of $a$ and $b$, scaled by a constant factor of $\abs{\bias_a - \bias_b}$. 

Next, we will work on determining $\costun$ so that no item will ever be left empty. Again, we wish to show that: 
$$ \abs{f(S_i) - \bias_{\type}} + \costun \geq \abs{f(S_i \setminus \bias_{\type})}+ 0$$
If we consider a reference player of type $A$, with $a$ players of type $A$ on item $i$ and $b$ of type $B$, then this becomes: 
$$\abs{\beta_a - \beta_b} \cd \frac{b}{a+b}+ \costun \geq \abs{\beta_a - \beta_b} \cd \frac{b}{a+b-1} + 0$$
$$\costun \geq  \abs{\beta_a - \beta_b} \cd b \cd \p{ \frac{1}{(a+b-1)} -\frac{1}{(a+b)}} $$

$$\costun \geq \abs{\beta_a - \beta_b} \cd b \cd  \frac{b}{(a+b-1)\cd (a+b)} $$

Next, we'll upper bound the term on the RHS. The RHS shrinks with $a$, so we can lower bound this by setting $a =1$. We know that $a\geq 1$ because we have assumed there is at least one player of type $A$ that wishes to move from the given item. The condition simplifies to: 
$$\costun \geq \abs{\beta_a - \beta_b}  \cd \frac{1}{1+b}$$
We similarly must have $b\geq 1$ (or else we're just modeling a single player of type $A$ move from one item to another). If we set $b=1$, then this goes to 1/2, which gives the desired condition. Intuitively, this tells us that we need that the cost of leaving something unlabeled is greater than half the distance between the two biases. 

Next, we will consider the case where the outcome function is equal to the median. Again, we wish to show that: 
$$ \abs{f(S_i) - \bias_{\type}} + \costun \geq \abs{f(S_i \setminus \bias_{\type})}+ 0$$
We will analyze multiple different cases for the potential outcome functions $f(S_i)$ and $\abs{f(S_i \setminus \bias_{\type})}$. Again, we will look from the perspective of a type $A$ agent on item $i$ considering moving to another item $j$ that is empty: 

\begin{itemize}
    \item $f(S_i) = \bias_a$ and $\abs{f(S_i \setminus \bias_{\type})} = \bias_a$. The inequality becomes: 
    $$0 + \costun \geq  0 + 0$$
    which is satisfied automatically. 
    \item $f(S_i) = \bias_a$ and $\abs{f(S_i \setminus \bias_{\type})} = \frac{1}{2} \cd (\bias_a + \bias_b)$. The inequality becomes: 
    $$0 +  \costun \geq 0.5 \cd \abs{\bias_a - \bias_b} +  0$$
    $$\costun \geq 0.5 \cd  \abs{\bias_a - \bias_b}$$
    \item $f(S_i) = \frac{1}{2} \cd (\bias_a + \bias_b) $ and $\abs{f(S_i \setminus \bias_{\type})} = \bias_b$. The inequality becomes: 
    $$0.5 \cd  \abs{\bias_a - \bias_b}  +  \costun \geq  \abs{\bias_a - \bias_b}+  0$$
    $$\costun \geq 0.5 \cd \abs{\beta_a - \beta_b}$$
    \item $f(S_i) = \frac{1}{2} \cd \bias_b$ and $\abs{f(S_i \setminus \bias_{\type})} = \bias_b$. The inequality becomes: 
    $$\abs{\bias_a - \bias_b} +  \costun \geq \abs{\bias_a - \bias_b} +  0$$
    which is always satisfied. 
\end{itemize}
The only inequality that isn't automatically satisfied is $\costun \geq 0.5 \cd  \abs{\bias_a - \bias_b}$, which is the same inequality as for the mean outcome function, and satisfied by the same reasoning. 

Finally, we will show that agents' incentives are independent of biases $\bias_a, \bias_b$. 

This proof comes almost immediately. \\
For mean outcome function, we can immediately see from the agent cost that agent strategy must be independent of biases. For every item with $a+b>0$, an agent's cost is solely a function of $a$ and $b$, scaled by a constant factor of $\abs{\bias_a - \bias_b}$. 

For median allocation, for any agent with bias $\bias_a$ and any set $S_i$, the outcome function has three possible values: $\bias_a$ (giving cost to agent $a$ of 0), $0.5 \cd (\bias_a + \bias_b)$ in the event of ties (giving cost to agent $a$ of $0.5 \cd \abs{\bias_a - \bias_b}$), or $\bias_b$ (giving cost to agent $a$ of $\abs{\bias_a - \bias_b}$). All of these are simply scaled values of $\abs{\bias_a - \bias_b}$, which means incentives are independent of the values $\bias_a, \bias_b$. 
\end{proof}

\medianneverstable*

\begin{proof}
Note that if we have $\nplayer = \nitem$, then the last criteria gives us $\nplayer_b < 0$, which isn't achievable. Thus, we will assume $\nplayer>\nitem$. \\

To start out with, we will describe a few arrangements where a deviation is always possible. For notation, we will use $\{a_1\geq b_1\}$ to denote an arrangement where there are at least as many type $A$ players as type $B$ on item 1, for example. We will only refer to items 1 and 2 for convenience, but these results hold for any labeled items. \\
\textbf{Case 1: }
\begin{equation*}
\{a_1 = b_1 - 1\}, \{a_2 \geq b_2 + 2\} \text{ or } \{b_1 = a_1-1\}, \{b_2 \geq a_2 + 2\}
\end{equation*}
This gives an arrangement where type $A$ loses by 1 on item 1 and wins by at least 2 on item 2.  This gives a deviation because any $a$ player from item 2 could move to label the item 1 and strictly reduce their cost (they now tie in the first and still win in the second). Similar reasoning holds for the other case: the type $B$ player from item 2 could move to label item 1 and strictly reduce their cost. \\
\textbf{Case 2:} 
\begin{equation*}
\{a_1 = b_1\}, \{a_2 \geq b_2+2\}  \text{ or } \{a_1 = b_1\}, \{b_2 \geq a_2 + 2\}
\end{equation*}
Here, the players tie on item 1 and type $A$ wins by at least 2 on item 2. This gives a deviation because any $a$ player from item 2 could move to label item 1 and strictly reduce their cost (they now win item 1 and still win item 2). Similar reasoning holds for the second case. \\
\textbf{Case 3: }
\begin{equation*}
\{a_1 \geq b_1 + 1\}, \{a_2 = b_2 + 1\} \text{ or } \{b_1 \geq a_1 + 1\}, \{b_2 = a_2 + 1\} 
\end{equation*}
Here, type $A$ wins by at least one on item 1 and wins by exactly one on item 2. This gives a deviation because type $B$ is losing in item 1, and can move to item 2 where it will tie (and still lose item 1). Similar reasoning holds for the second case. 

Next, we will show that if we have any arrangement satisfying the preconditions ($\nplayer_a + \nplayer_b \leq 2\nitem, \nitem < \nplayer_a, \nplayer_b < \nplayer_a- \nitem$), then at least one of these cases must occur, meaning that the arrangement must be unstable.

First, let's suppose that we have at least one tie somewhere: $\{a_1 = b_1\}$. By Case 2, we know that we can't have any player types win by 2 elsewhere (or else one of the agents could move and win at item 1, while still winning elsewhere). That means that we need: 
$$b_i \leq a_i + 1 \text{ and } a_i \leq b_i + 1 \ \forall i \ne 1 $$
We will show that satisfying the second condition is impossible. First, summing over all of the items gives: 
$$\sum_{i \ne 1} a_i \leq \sum_{i \ne 1} b_i + (\nitem-1)$$
$$\nplayer_a - a_1 \leq \nplayer_b - b_1 + (\nitem-1)$$
Recall that $a_1 = b_1$ (there's a tie), so this reduces to: 
$$\nplayer_a \leq \nplayer_b + (\nitem-1)$$
$$\nplayer_a - (\nitem-1) \leq \nplayer_b$$
However, we also have $\nplayer_b < \nplayer_a- \nitem$, so this becomes: 
$$\nplayer_a - (\nitem-1) \leq \nplayer_b < \nplayer_a- \nitem$$
$$1 < 0$$
which is a contradiction. This implies that no arrangement where players tie can be stable (and satisfy the preconditions). 

Next, let's consider the case where we have $a_1 = b_1-1$ (type $A$ loses item 1 by exactly 1). Case 1 tells us that we cannot have $a_i \geq b_i+2$ elsewhere (cannot have that type $A$ wins by at least two elsewhere), so we must have $a_i \leq b_i + 1$. Again, we can sum: 
$$\sum_{i\ne 1} a_i \leq \sum_{i \ne 1} b_i + 1$$
$$\nplayer_a -a_1 \leq \nplayer_b - b_1 + (\nitem-1)$$
Substituting in for $a_1 = b_1-1$ gives us: 
$$\nplayer_a - b_1+1 \leq \nplayer_b - b_1 + \nitem-1$$
$$\nplayer_a \leq \nplayer_b + \nitem-2$$
$$\nplayer_a - \nitem + 2 \leq \nplayer_b $$
which again contradicts $\nplayer_b < \nplayer_a - \nitem$. 

Finally, let's consider where we have $b_1 = a_1-1$ or $b_1 + 1 = a_1$. By Case 1, we must have $b_i \leq a_i + 1$ for all $i \ne 1$. However, we will show that this implies a contradiction with other stability analayses. 

First, let's consider the case where $b_i = a_i + 1$ for all $i \neq \nitem$, so type $B$ wins every item by one except the first one, which it loses by exactly one. We will show that this violates the preconditions of this lemma: 
$$b_i = a_i +1 \forall i \ne 0$$
$$\sum_{i\ne 1} b_i = \sum_{i \ne 1} a_i + 1$$
$$\nplayer_b - b_1 = \nplayer_a - a_1 + \nitem-1$$
$$\nplayer_b - (a_1-1) = \nplayer_a - a_1 + \nitem-1$$
$$\nplayer_b + 1 = \nplayer_a + \nitem-1$$
$$\nplayer_b = \nplayer_a + \nitem-2$$
If we combine this with our  $\nplayer_b < \nplayer_a - \nitem $ condition, we get: 
$$\nplayer_a + \nitem-2 < \nplayer_a - \nitem$$
$$2 \nitem <2$$
$$\nitem<1$$
which cannot be satisfied because we need at least one item. 

Next, we'll consider the case where $b_i < a_i +1$ for at least one item. From our previous analysis of Case 2, we know that we cannot have an exact tie ($a_i = b_i$). 

This implies that we must have at least one item $i$ such that $b_i \leq a_i -1$, along with item 1 which has $b_1 = a_1-1$.  Equivalently, we can write this as $a_1 = b_1 + 1, a_i \geq b_i+1$. However, this exactly implies the condition in Case 3, which is also unstable. 
\end{proof}

\medianmstable*
\begin{proof}
We will prove this result constructively by producing an algorithm that always arrives at a stable arrangement. 

Informally, this algorithm works by putting 2 type $B$ players on every item, stopping once either a) item $\nitem-1$ is reached, or b) all of the type $B$ players have been assigned, or c) there are 3 type $B$ players left (which are then all assigned to the current item). Then, the algorithm places at least 2 type $A$ players on each item, again stopping once a) item $\nitem$ is reached, or b) all of the type $A$ players have been allocated, or c) there are 3 type $A$ players left. This is a stable arrangement because for each item, each type wins by at least 2, so no single agent acting alone can change the outcome. For a formal description, see Algorithm \ref{alg:allocatemore}. 

If $\nplayer_a + \nplayer_b = 2 \cd \nitem$ and both $\nplayer_a, \nplayer_b$ are even, then this algorithm will put exactly 2 agents on each item. If $\nplayer_a + \nplayer_b \geq 2 \cd \nitem +1$ and both $\nplayer_a, \nplayer_b$ are even, then this will put at least 2 agents on each item. If $\nplayer_a + \nplayer_b \geq 2 \cd \nitem +1$ and both $\nplayer_a, \nplayer_b$ are odd, then $\frac{\nplayer_b -1}{2}$ items will have type $B$ agents and $\frac{\nplayer_a - 1}{2}$ items will have type $A$ agents. In total, this covers $\frac{\nplayer_a+ \nplayer_b - 2}{2}  = \frac{\nplayer_a +\nplayer_b}{2} - 1  \geq \frac{2 \cd \nitem + 1}{2} - 1 = \nitem$ items, as desired. 
\end{proof}

\medianmstable*
\begin{proof}
We will prove this result constructively by producing an algorithm that always arrives at a stable arrangement. 

Informally, this algorithm works by putting 2 type $B$ players on every item, stopping after a) item $\nitem-1$ is reached, or b) all of the type $B$ players have been assigned, or c) there are 3 type $B$ players left (which are then all assigned to the current item). Then, the algorithm places at least 2 type $A$ players on each remaining item, stopping after a) item $\nitem$ is reached, or b) all of the type $A$ players have been allocated, or c) there are 3 type $A$ players left. This is a stable arrangement because for each item, each type wins by at least 2, so no single agent acting alone can change the outcome. Because it is impossible for any single agent to change the outcome, this arrangement is stable for all positive weights over items. For a formal description, see Algorithm \ref{alg:allocatemore}. 

If $\nplayer_a + \nplayer_b = 2 \cd \nitem$ and both $\nplayer_a, \nplayer_b$ are even, then this algorithm will put exactly 2 agents on each item. If $\nplayer_a + \nplayer_b \geq 2 \cd \nitem +1$ and both $\nplayer_a, \nplayer_b$ are even, then this will put at least 2 agents on each item. If $\nplayer_a + \nplayer_b \geq 2 \cd \nitem +1$ and both $\nplayer_a, \nplayer_b$ are odd, then $\frac{\nplayer_b -1}{2}$ items will have type $B$ agents and $\frac{\nplayer_a - 1}{2}$ items will have type $A$ agents. In total, this covers $\frac{\nplayer_a+ \nplayer_b - 2}{2}  = \frac{\nplayer_a +\nplayer_b}{2} - 1  \geq \frac{2 \cd \nitem + 1}{2} - 1 = \nitem$ items, as desired. 
\end{proof}

\medianstableequal*

\begin{proof}
First, we note that if $\nplayer_b =0$, then every arrangement is stable because every arrangement has only players of type $A$ on them. Thus, we will require $\nplayer_b\geq 1$. 

We will show constructively that it is possible to create an arrangement satisfying the following criteria: 
\begin{itemize}
    \item For every item where there is more than 1 agent, type $A$ and type $B$ tie exactly. 
    \item Every other item has exactly one agent, which can be either type $A$ or type $B$. 
\end{itemize}

This type of construction is stable by the following reasoning: 
\begin{itemize}
        \item None of the single agents can move (they can't leave an item empty). 
        \item No agent on an item with multiple agents wishes to leave - they would go from winning a single item and losing another, to losing on that item and tying on another, which gives equal costs. 
\end{itemize}

Algorithm \ref{alg:allocatefewer} considers this case. Informally, we will describe how it works. 

If $\nplayer_a + \nplayer_b \leq \nitem$, then we place at most one agent on each item, which satisfies the construction criteria. 

Next, we consider the case where $\nplayer_a + \nplayer_b -\nitem$ is even. We calculate  $\frac{\nplayer_a + \nplayer_b-\nitem}{2} = x$ and place $x$ of type $A$ and type $B$ players each on item 0. By the assumptions of this lemma, we know that $\nplayer_b \geq \nplayer_a-\nitem$, which means that $x = \frac{\nplayer_a + \nplayer_b - \nitem}{2} \leq \frac{2 \nplayer_a- 2 \cd \nitem}{2} = \nplayer_a - \nitem >0$. Then, we place exactly 1 of type $A$ and type $B$ on item 1. For every other item, have $\nplayer_a - x - 1$ with exactly one type $A$ agent, exactly $\nplayer_b - x -1$ with exactly one type $B$ agent. The total number of items: 
$$1 + 1 + \nplayer_a - x -1 + \nplayer_b - x -1  $$
$$=\nplayer_a + \nplayer_b -2x = \nplayer_a + \nplayer_b - \nplayer_a -\nplayer_b + \nitem = \nitem$$
as desired.

Next, if $\nplayer_a + \nplayer_b - \nitem$ is odd, we know that $\nplayer_a+ \nplayer_b -(\nitem-1)$ is even. We then calculate $\frac{\nplayer_a+ \nplayer_b+1-\nitem}{2} = x$. 
Then, set $a_0 = x, b_0 = x$ (we address the case where $\nplayer_b <x$ at the end of this proof). Every other item has exactly one agent, with $\nplayer_a-x$ of them having type $A$, and $\nplayer_b -x$ having type $B$. Note that the total number of agents adds up to the right amount and that (by construction) every item is labeled. 1 item has many players, $\nplayer_a- x$ have type $A$ only, $\nplayer_b-x$ have type $B$ only, and together this sums to: 
$$1 + \nplayer_a-x + \nplayer_b -x  = 1 + \nplayer_a+ \nplayer_b - 2 \cd x $$
$$= 1 + \nplayer_a+ \nplayer_b - (\nplayer_a+ \nplayer_b - (\nitem-1)) = \nitem$$

Finally, we consider the case where $x = \frac{\nplayer_a + \nplayer_b +1 - \nitem}{2}>\nplayer_b$, which means $\nplayer_a + 1 - \nitem > \nplayer_b$. By assumptions of this lemma, we know that $\nplayer_b \geq \nplayer_a - \nitem$. Taken together, this tells us that: 
    $$\nplayer_a - \nitem \leq \nplayer_b < \nplayer_a  - \nitem+1$$
    This can only be satisfied by setting $\nplayer_b = \nplayer_a - \nitem$. However, this means that: 
    $$\nplayer_a + \nplayer_b - \nitem = 2 \cd \nplayer_a - 2 \cd \nitem 
 = 2 \cd (\nplayer_a - \nitem)$$
    However, this conflicts with the assumption that $\nplayer_a + \nplayer_b - \nitem$ is odd, which means this situation can never occur. 

\end{proof}

\meanbowl*
\begin{proof}
In order to show this, we will look at a relaxed (continuous) version of this problem. In this relaxed version of the problem, we will assume that, instead of agents coming in integer units, they can be allocated fractionally across items.  

The cost to player of type $A$ is given by 
$\sum_{i=1}^{\nitem} \frac{b_i}{a_i+b_i}$. 
The partial derivative of this with respect to $a_i$ is: 
$$\frac{\partial}{\partial a_i} \sum_{i=1}^{\nitem} \frac{b_i}{a_i+b_i} = - \frac{b_i}{(a_i + b_i)^2}$$
In order for us to be at a stable point, we need that the derivative wrt $a_i$ must be equal to the derivative wrt $a_j$ for any $i \ne j$ and $\sum_{i=1}^{\nitem} a_i = \nplayer_a$ (identical criteria for $B$). If the first is not satisfied, then type $A$ could strictly reduce its cost by changing its allocations between items, and if the second is not satisfied, then type $A$ could again reduce its cost by allocating more agents onto items. \\

We can exactly achieve this by setting $a_i = {\nplayer_a}{\nitem}, b_i = \frac{\nplayer_b}{\nitem}$ for all players. This automatically gives us $\sum_{i = 1}^{\nitem} a_i = \nplayer_a, \sum_{i = 1}^{\nitem} = \nplayer_b$, and also sets derivatives equal because each item has identical proportions of $a_i, b_i$.  
\end{proof}

\meanbigunstable*

\begin{proof}
We will set parameters as follows: \\
$\nitem=2$ items, and $\nplayer_a', \nplayer_b'$ as any pair of odd numbers separated by exactly 2 that are at least as large as $\nplayer_a, \nplayer_b$, or:
\begin{itemize}
    \item $\nplayer_a' = \nplayer_b'+2$
    \item $\nplayer_a', \nplayer_b'$ both odd
    \item $\nplayer_a' \geq \nplayer_a, \nplayer_b' \geq \nplayer_b$
\end{itemize}
Our goal will be to show that there is no stable arrangement of $\nplayer_a', \nplayer_b'$ players onto $\nitem=2$ items. From Theorem \ref{thrm:meanclose}, we know that any stable arrangement (if it exists) must be close to even, or here, we must have $\abs{a_i - \frac{\nplayer_a'}{2}} \leq 1, \abs{b_i - \frac{\nplayer_a'}{2}} \leq 1$. For odd $\nplayer_a', \nplayer_b'$, this means that we must have that we must have $\abs{a_1-a_2} = \abs{b_1-b_2} = 1$, or that the number of type A players must be $(\nplayer_a'+1)/2, (\nplayer_a'-1)/2$ and the number of type $B$ players must be $(\nplayer_b'+1)/2, (\nplayer_b'-1)/2$. There are exactly two ways that they could be arranged on the two items: 
\begin{enumerate}
    \item Anti-correlated: type $A$ has more of its agents on item 1, and type $B$ has more of its agents on item 2, or: $(a_1, b_1) = ((\nplayer_a'+1)/2, (\nplayer_b'-1)/2)$ and $(a_2, b_2) = ((\nplayer_a'-1)/2, (\nplayer_b'+1)/2)$
    \item Correlated: both type $A$ and type $B$ have more agents on item 1, or: $(a_1, b_1) = ((\nplayer_a'+1)/2, (\nplayer_b'+1)/2)$ and $(a_2, b_2) = ((\nplayer_a'-1)/2, (\nplayer_b'-1)/2)$. 
\end{enumerate}

We will show that both of these arrangements are unstable by showing that at least one player wishes to change which item they are labeling. 

\textbf{Case 1: Anti-correlated}:\\
In this setting, we will show that an agent of type $A$ currently labeling item 1 can reduce its cost by moving to label item 2. \\
Before moving, type $A$'s cost is: 
$$\frac{b_1}{a_1+b_1} + \frac{b_2}{a_2 + b_2} = \frac{\nplayer_b'-1}{\nplayer_a' + \nplayer_b'} + \frac{\nplayer_b'+1}{\nplayer_a' + \nplayer_b'}$$
After moving, type $A$'s cost becomes: 
$$\frac{b_1}{a_1+b_1-1} + \frac{b_2}{a_2 + b_2+1} = \frac{\nplayer_b'-1}{\nplayer_a' + \nplayer_b'-1} + \frac{\nplayer_b'+1}{\nplayer_a' + \nplayer_b'+1}$$
Our goal is to show that: 
$$\frac{\nplayer_b-1}{\nplayer_a' + \nplayer_b'} + \frac{\nplayer_b'+1}{\nplayer_a' + \nplayer_b'} > \frac{\nplayer_b'-1}{\nplayer_a' + \nplayer_b'-1} + \frac{\nplayer_b'+1}{\nplayer_a' + \nplayer_b'+1}$$
Rearranging: 
$$(\nplayer_b'+1) \cd \p{\frac{1}{\nplayer_a' + \nplayer_b'+1} - \frac{1}{\nplayer_a' + \nplayer_b'}} $$
$$> (\nplayer_b'-1) \cd \p{\frac{1}{\nplayer_a' + \nplayer_b'-1} - \frac{1}{\nplayer_a' + \nplayer_b'}}$$
Simplifying: 
$$(\nplayer_b'+1) \cd \frac{1}{(\nplayer_a' + \nplayer_b'+1) \cd (\nplayer_a' + \nplayer_b')} $$
$$> (\nplayer_b'-1)\cd \frac{1}{(\nplayer_a' + \nplayer_b'-1) \cd (\nplayer_a' + \nplayer_b')}$$
Or: 
$$\frac{\nplayer_b'+1}{\nplayer_a' + \nplayer_b'+1} > \frac{\nplayer_b'-1}{\nplayer_a' + \nplayer_b'-1}$$
Note that we can write this as: 
$$\frac{x}{a+x} > \frac{y}{a+y}$$
$$a \cd x + x \cd y > a \cd y + x \cd y$$
$$ \nplayer_b' + 1= x > y = \nplayer_b'-1$$
which is satisfied. \\

\textbf{Case 2: Correlated}: \\
In this setting, we will show that an agent of type $B$ can reduce their cost by moving from item 1 to item 2. \\
Before moving, type $B$'s cost is: 
$$\frac{\nplayer_a' + 1}{\nplayer_a' + \nplayer_b'+2} + \frac{\nplayer_a'-1}{\nplayer_a' + \nplayer_b'-2}$$
After moving, type $B$'s cost is: 
$$\frac{\nplayer_a' + 1}{\nplayer_a' + \nplayer_b'+1} + \frac{\nplayer_a'-1}{\nplayer_a' + \nplayer_b'-1}$$
Our goal is to show that: 
$$ \frac{\nplayer_a' + 1}{\nplayer_a' + \nplayer_b'+2} + \frac{\nplayer_a'-1}{\nplayer_a' + \nplayer_b'-2} > \frac{\nplayer_a' + 1}{\nplayer_a' + \nplayer_b'+1} + \frac{\nplayer_a'-1}{\nplayer_a' + \nplayer_b'-1}$$
Rearranging: 
$$ (\nplayer_a'-1) \cd \p{\frac{1}{\nplayer_a' + \nplayer_b'-2} - \frac{1}{\nplayer_a' + \nplayer_b'-1}} $$
$$> (\nplayer_a'+1) \cd \p{\frac{1}{\nplayer_a' + \nplayer_b'+1} - \frac{1}{\nplayer_a'+ \nplayer_b'+2}}$$
$$\frac{\nplayer_a'-1}{(\nplayer_a' + \nplayer_b'-2) \cd (\nplayer_a' + \nplayer_b'-1)} $$
$$> \frac{\nplayer_a'+1}{(\nplayer_a' + \nplayer_b'+1) \cd (\nplayer_a' + \nplayer_b'+2)}$$
For simplicity, let's rewrite this as $x = \nplayer_a', y = \nplayer_a' + \nplayer_b'$. Then, the quantity we're trying to prove becomes: 
$$\frac{x-1}{(y-2) \cd (y-1)} > \frac{x+1}{(y+1) \cd (y+2)}$$
$$(x-1) \cd (y+1) \cd (y+2) > (x+1) \cd (y-2) \cd (y-1)$$
Expanding: 
$$(x-1) \cd (y^2 + 3 \cd y + 2) > (x+1) \cd (y^2 -3\cd y + 2)$$
$$6 \cd x \cd y > 4 + 2 \cd y^2$$
$$3 \cd x \cd y > 2 + y^2$$
$$y \cd (3 \cd x - y) > 2 $$
Next, substituting back in for $x, y$ gives us: 
$$(\nplayer_a' + \nplayer_b') \cd (3 \cd \nplayer_a - \nplayer_a' - \nplayer_b')>2$$
$$(\nplayer_a' + \nplayer_b') \cd (2 \cd \nplayer_a - \nplayer_b')>2$$
Finally, we recall that we set $\nplayer_a' = \nplayer_b'+2$. Substituting in for this gives: 
$$(2\cd \nplayer_b'+2) \cd (2 \cd \nplayer_b' + 4-\nplayer_b)>2$$
$$(2\cd \nplayer_b'+2) \cd (\nplayer_b' + 4)>2$$
which is satisfied because $\nplayer_b'\geq 1$. 
\end{proof}

\meanclose*
\begin{proof}
In order to prove this result, we will show that any arrangement that is not \enquote{close} to proportional must have at least one agent that wishes to move. An arrangement is \enquote{close} to proportion if: 
$$\abs{a_i - \nplayer_a/\nitem} \leq 1, \abs{b_i - \nplayer_b/\nitem} \leq 1 \ \forall i \in [\nitem]$$
For any arrangement falling outside of these parameters, we will give a move that reduces cost for at least one player type. At a high level, this will involve finding a \enquote{large} item (with more agents) and a \enquote{small} item (with fewer agents). We will then pick whichever agent type is more represented in the larger item, and move exactly one agent to the smallest item. We will show that, almost always, this strictly reduces cost for the agent type that moves. In the cases where such a move would not strictly reduce cost, we show that all items must be \enquote{close} to proportional. 

First, we will consider all pairs of items $j, k$ such that $a_j + b_j \geq a_k + b_k$. One feature of this pair that we will consider is the \emph{gap in type prevalence}, which is given by: 

$$\text{ type A: } \frac{a_j}{a_j + b_j} - \frac{a_k}{a_k + b_k} \quad \text{ type B: }  \frac{b_j}{a_j + b_j} - \frac{b_k}{a_k + b_k}$$
Note that if one type has positive gap in type prevalence, then the other has negative gap in type prevalence, because: 
$$\frac{a_j}{a_j + b_j} - \frac{a_k}{a_k + b_k} $$
$$ = -1 + \frac{a_j}{a_j + b_j} + 1- \frac{a_k}{a_k + b_k} = \frac{-b_j}{a_j + b_j} + \frac{b_k}{a_k + b_k}$$
WLOG, we will assume that type $A$ has positive or 0 gap in type prevalence. If the gap in type prevalence is 0 (both items have exactly equal proportions of player types), then we will again WLOG assume that type $A$ makes up a larger share of players, or $\frac{a_j}{a_j + b_j} = \frac{a_k}{a_k + b_k} \geq \frac{b_j}{a_j + b_j} = \frac{b_k}{a_k + b_k}$. Given this assumption, we will show that players of type $A$ could always reduce its cost by moving a single agent from item $j$ to item $k$ (unless all items are \emph{close} to proportional).

\textbf{Costs:} 
Type $A$'s cost is: 
$$\frac{b_j}{a_j + b_j} + \frac{b_k}{a_k + b_k}$$
Its cost after moving a single agent from $a_j$ is given by: 
$$\frac{b_j}{a_j + b_j-1} + \frac{b_k}{a_k + b_k+1}$$
So, its cost decreases whenever: 
$$\frac{b_j}{a_j + b_j} + \frac{b_k}{a_k + b_k} -\frac{b_j}{a_j + b_j-1} - \frac{b_k}{a_k + b_k+1} >0$$
Or: 
\begin{equation}\label{eq:stableeq}
b_k \cd (a_j+ b_j) \cd (a_j + b_j -1) > b_j \cd (a_k + b_k) \cd (a_k + b_k +1)
\end{equation}
Equation \ref{eq:stableeq} is the central condition we will be studying in this proof. First, we will present several sufficient conditions for when  Equation \ref{eq:stableeq} is satisfied, so a player wishes to move. Then, we will show that whenever none of those sufficient conditions are satisfied, all items must be \enquote{close} to proportional. 

\textbf{Sufficient condition 1: items differ by at least 2, positive gap in type prevalence}
There are a few conditions where we can immediately see that Equation \ref{eq:stableeq} is satisfied. First, we divide this equation into two separate components. The first is given by: 
$$b_k \cd (a_j+ b_j) \geq b_j \cd (a_k + b_k)$$
\begin{equation}\label{eq:overprop}
\frac{b_k}{a_k + b_k} \geq \frac{b_j}{a_j + b_j}
\end{equation}
By prior reasoning, this is satisfied by the assumption that $\frac{a_j}{a_j + b_j} \geq \frac{a_k}{a_k + b_k}$ as given by type $A$ having positive gap in type prevalence. \\ 
Next, we will consider the second component of Equation \ref{eq:stableeq}: 
$$a_j + b_j -1 \geq a_k + b_k +1$$
which is satisfied exactly whenever: 
\begin{equation}\label{eq:more1}
a_j + b_j \geq a_k + b_k +2
\end{equation}
Note that this is \emph{not} required by how we selected items (all we require is $a_j + b_j \geq a_k + b_k$). However, in the event that both Equation \ref{eq:more1} holds, and either Equation \ref{eq:overprop} or Equation \ref{eq:more1} is satisfied strictly, then Equation \ref{eq:stableeq} hold strictly and type $A$ players have an incentive to move from item 1 to item 2. \\
Next, we will consider other cases where Equation \ref{eq:more1} and \ref{eq:overprop} are not both strictly satisfied, and yet a type $A$ player still wishes to move from item $j$ to $k$. 

\textbf{Sufficient condition 2: Larger and smaller item have same number of players, $a_j \geq a_k + 2$} First, we will consider the case where $a_j + b_j = a_k + b_k$ (the two items have equal numbers of players).  Equation \ref{eq:stableeq} simplifies to 
$$b_k \cd (a_j + b_j -1) > b_j \cd (a_k + b_k +1)$$
We will show that so long as $a_j \geq a_k + 2$, then this inequality must hold. Note that this condition implies $b_k \geq b_j + 2$. Substituting in for $a_j + b_j = a_k + b_k$ and $b_k \geq b_j + 2$ tells us that we wish to show: 
$$(b_j +2) \cd (a_j + b_j -1) > b_j \cd (a_j + b_j + 1)$$
$$b_j \cd (a_j + b_j) - b_j + 2 \cd (a_j + b_j) - 2> b_j \cd (a_j + b_j) + b_j$$
$$ - b_j + 2 \cd (a_j + b_j) - 2> b_j$$
$$2 \cd a_j + 2 \cd b_j > 2 \cd b_j + 2$$
$$a_j \geq 1$$
Because $a_j \geq a_k + 2$, we know that $a_j \geq 1$, as desired. \\
When this sufficient condition is not met, we have that $a_j \leq a_k + 1$. Because we assume that type $A$ has positive gap in type prevalence, this means that $a_j = a_k + 1$ and $b_j = b_k-1$ in order to ensure $a_j + b_j = a_k + b_k$. 

\textbf{Sufficient condition 3: Larger item has exactly 1 more agent than smaller,  $a_j \geq a_k + 2$} 
Next, we will consider the case where the larger item has exactly 1 more player than the smaller item, or $a_j + b_j = a_k + b_k + 1$. Then, the condition for a player of type A wanting to move (Equation \ref{eq:stableeq}) becomes: 
$$b_k \cd (a_k + b_k + 1) \cd (a_k + b_k) > b_j \cd (a_k + b_k) \cd (a_k + b_k + 1)$$
which simplifies down to: 
$$b_k > b_j$$
If we have $a_j \geq a_k + 2$, then we must have: 
$$a_k + b_k = a_j + b_j -1 > a_k + 2 + b_j -1 = a_k + b_j + 1$$
which implies that: 
$$b_k \geq b_j + 1$$
as desired.\\ 
When this sufficient condition is not met, we have that $a_j \leq a_k + 1$. Because we assume that type $A$ has positive gap in type prevalence, this means that $a_j = a_k + 1$ and $b_j = b_k$ in order to ensure $a_j + b_j = a_k + b_k+1$. \\
\textbf{In absence of sufficient conditions, all items are \enquote{close} to proportional}\\
Finally, we will consider the case where none of the sufficient conditions hold. By taking the negative of previous cases, we know that: 
\begin{enumerate}
    \item By sufficient condition 1, there is no pair of items $j, k$ with $a_j + b_j > a_k + b_k + 2$. 
    \item By sufficient condition 1 for any pair of items with $a_j + b_j = a_k + b_k + 2$, we must have zero gap in type prevalence. 
    \item By sufficient condition 2, for any pair of items with $a_j + b_j = a_k + b_k$, we must have $a_j = a_k + 1, b_j = b_k -1$. 
    \item By sufficient condition 3, for any pair of items with $a_j + b_j = a_k + b_k+1$, we must have $a_j = a_k + 1, b_j = b_k$. 
\end{enumerate}
Each of these cases (except for item 2) results in an arrangement where for any pair of items $j, k$, we have
$$\abs{a_j - a_k} \leq 1 \quad \abs{b_j - b_k} \leq 1$$
If we are also able to show that this holds for item 2 (when $a_j + b_j = a_k + b_k +2$), then we will know that \emph{any} pair of items differs in $a, b$ by at most 1. This would mean that every item has either $a_i = a^*$ or $a^*+1$ and $b_i = b^*$ or $b^*+1$. The average number of players per items must lie somewhere between $[a^*, a^*+1], [b^*, b^*+1]$, which means that every item is \enquote{close} to proportional.  

We will conclude our proof by showing that examining item 2 and showing that we must also have $\abs{a_j - a_k} \leq 1,  \abs{b_j - b_k} \leq 1$. 

In this case, we assume that $a_j + b_j = a_k + b_k + 2$. Note that by our analysis of Case 1, a type $A$ player wishes to move automatically if Equation \ref{eq:overprop} holds strictly. In this case, we will  assume that it does \emph{not} hold strictly and derive additional conditions on $a_j, b_j, a_k, b_k$. Equation \ref{eq:overprop} becomes an equality: 
$$\frac{b_k}{a_k + b_k} = \frac{b_j}{a_j + b_j}$$
This implies that items $j, k$ have exactly the same proportion of type $A$ and type $B$ players (a zero gap in type prevalence). This means that we can write $a_j = c \cd b_j, a_k = c \cd b_k$. In how we selected items, we assumed WLOG that in the case that the gap in type prevalence is 0, players of type $A$ make up a greater share of players on both items. This implies that  $c\geq 1$. 
We can then write: 
$$a_j + b_j = a_k + b_k +2$$
Substituting in for $a_j = c\cd b_j$ gives: 
$$b_j \cd (c+1) = b_k \cd (c+1) + 2$$
$$b_j = b_k + \frac{2}{c+1}$$
The only value $c$ for which this results in an integer $b_j, b_k$ is $c=1$, which implies $b_j = b_k + 1$, $a_k = b_k$ and $a_j = b_j = a_k + 1$, or $\abs{a_j - a_k} \leq 1,  \abs{b_j - b_k} \leq 1$ as desired. 
\end{proof}

\meanmisallocate*
\begin{proof}
This is a direct consequence of Theorem \ref{thrm:meanclose}:
\begin{equation*} \sum_{i \in [\nitem]}\abs{\nplayer_a/\nitem - a_i} + \abs{\nplayer_b/\nitem - b_i} \leq \sum_{i \in [\nitem]} 2 = 2 \cd \nitem \end{equation*}
\end{proof}

\medianmisallocate*
\begin{proof} 
We rely on the stable arrangements found in the proof of Lemma \ref{lem:medianstableequal}. For $\nplayer_a + \nplayer_b \geq 2 \cd \nitem$, the arrangement starts by placing 2 of each of type $B$ players on each item, up until either we run out of type $B$ players or reach the $\nitem-1$st item. We then place the type $A$ on the remaining items. In the case that $\frac{\nplayer_b}{2} \geq \nitem-1$, this implies that a single item has $\nplayer_a$ type $A$ agents. As compared to an even split of players over items, this means that misallocated effort is lower bounded by 
$\nplayer_a - \frac{\nplayer_a}{\nitem} = \nplayer_a \cd \p{1- \frac{1}{\nitem}}$. 
Because we have required $\nplayer_a \geq \nplayer_b$ and $\nitem \geq 2$, this bound is at least: $0.5 \cd \nplayer \cd \p{1 - 0.5} = 0.25 \cd \nplayer$.
\end{proof}

\twostable*
\begin{proof} 
Consider two players with bias $\beta_a, \beta_b$. The player of type $A$ get lower cost when competing in separate items whenever: 
\begin{align*}
 \abs{f(\{\beta_a, \beta_b\}) - \beta_a} +\costun & \leq 0 + \abs{f(\beta_b) - \beta_a}\\
 \costun & \leq 0.5 \cd \abs{\beta_b - \beta_a}
\end{align*}
On the other hand, the player of type $B$ gets lower cost when competing in separate items whenever: 
\begin{align*}
\abs{f(\{\beta_a, \beta_b\}) - \beta_b} + \costun  & \leq \abs{f(\beta_a) - \beta_b} + 0\\
0.5 \cd \abs{\beta_b - \beta_a} + \costun & \leq \abs{\beta_b - \beta_a}
\end{align*}
These terms are exactly equivalent, proving the result. 
\end{proof}

\medianunstable*

\begin{proof}
Set parameters as follows: 

\begin{itemize}
    \item $\nitem$ items. 
    \item $N$ players, with given biases: 1 with bias $1$, and $N-1$ with bias $-0.5$. 
    \item Cost of 0.3 for leaving an item empty. 
\end{itemize}

First, we will note that any arrangement where 3 or more players are labeling a single item (item $i$) must fail to be stable. Because $N<\nitem$, we know that every arrangement must leave at least one item empty, for a cost of 0.3. 

\begin{itemize}
    \item First, suppose that the item in question has only players with bias $-0.5$ labeling it. Then, any player can leave item $i$ and instead label an empty item $j$. The median of item $i$ remains the same and the median of $j$ becomes $-0.5$ (eliminating empty cost), which reduces the cost for players with bias $-0.5$. 
    \item Next, consider the other case where the item in question is labeled by the player with bias 1. We know that there are at least 2 players with bias $-0.5$, since there are at least 3 players total. If the player with bias 1 leaves $i$ and instead labels $j$, then the median remains -0.5. However, again, this now eliminates the penalty for leaving $j$ empty. 
\end{itemize}
Now, we consider all arrangements where there are no more than 2 players labeling an item. We similarly know that if there are 2 players of bias $-0.5$ labeling an item, then they can reduce their cost by having one leave to label an empty item (maintaining the same median and eliminating empty cost). Therefore, we can focus on how the player with bias $1$ is arranged. We will find that considering only the first 3 items will be sufficient. For conciseness, we will use the notation $\{1\}$ $\{-0.5\}, \{-0.5\}$ to mean that three players are all labeling a single item each (one player with bias 1, two players with bias $-0.5$). 

\begin{itemize}
    \item $\{1\}$ $\{-0.5\}, \{-0.5\}$ goes to $\{1, -0.5\}, \{\}, \{-0.5\}$. Originally, the -0.5 players have cost $1.5 + 0 + 0 = 1.5$. The $-0.5$ player wishes to move, which produces medians $0.25, -0.5, -0.5$, which gives it cost become $0.75 + 0.3 + 0 = 0.75 <1.5$. 
    \item $\{1, -0.5\}, \{\}, \{-0.5\}$ goes to $\{1, -0.5, -0.5\}, \{\}, \{\}$. The $-0.5$ player wishes to move again. Before, the medians are $0.25, \text{n/a}, -0.5$, giving it a cost of $0.75 + 0.3 + 0 = 1.05$. After it moves, the medians are $-0.5, \text{n/a}, \text{n/a}$, which gives it a cost of $0 + 0.3\cd 2 = 0.6$.
    \item $\{1, -0.5, -0.5\}, \{\}, \{\}$ goes to $\{-0.5, -0.5\}, \{1\}, \{\}$. The 1 player wishes to move. Before, the medians are $-0.5, \text{n/a}, \text{n/a}$, which gives it a cost of $1.5 + 0.3 \cd 2 = 2.1$. After it moves, the medians are $-0.5, 1, \text{n/a}$, which gives it a cost of $1.5 + 0 + 0.3 = 1.8 < 2.1$. 
    \item $\{-0.5, -0.5\}, \{1\}, \{\}$ goes to $\{-0.5,\}, \{1\}, \{-0.5\}$. The $-0.5$ player wishes to move. Before, the medians are $-0.5, 1, \text{n/a}$, which gives it cost $0 + 1.5 + 0.3 = 1.8$. After it moves, the medians are $-0.5, 1, -0.5$, which gives it cost $0 + 1.5 + 0 = 1.5 < 1.8$.
\end{itemize}
\end{proof}

\meanunstable*

\begin{proof}
Set parameters as follows: 

\begin{itemize}
    \item $\nitem$ items.  
    \item $N$ players, with given biases: 1 with bias $1$, and $N-1$ with bias $-0.5$. 
    \item Cost of $\costun \in (0.125, 0.25)$ for leaving an item empty. 
\end{itemize}

Because $N<\nitem$, we know that every arrangement must leave at least one item empty, for a cost of $\costun$. First, we will consider the case where at least one item $j$ has 3 or more agents on it. 

\begin{itemize}
    \item First, suppose that the item with 3 or more players labeling it has only players with bias $-0.5$ labeling it. Then, any player can leave item $i$ and instead label an empty item $j$. The mean of item $i$ remains the same and the mean of $j$ becomes $-0.5$ (eliminating empty cost), which reduces the cost for players with bias $-0.5$. 
    \item Next, consider the other case where the item in question is labeled by the player with bias 1. Say that there are $a_i$ players with bias -0.5. We know that there $a_i \geq 2$ since there are at least 3 players total. 
    \begin{itemize}
        \item Currently, the cost to players of type $A$ is: 
    $$\abs{\frac{-0.5 \cd a_i + 1}{a_i + 1} + 0.5} + \costun $$
    $$= \abs{\frac{-0.5 \cd a_i + 1 + 0.5 \cd a_i + 0.5}{a_i + 1}} + \costun =  \frac{1.5}{a_i+1} + \costun$$
    If one player of type $A$ goes to label the empty item, the cost becomes: 
    $$\abs{\frac{-0.5 \cd (a_i-1) + 1}{a_i} + 0.5} $$
    $$= \abs{\frac{-0.5 \cd a_i + 0.5 + 1 + 0.5 \cd a_i}{a_i}} = \frac{1.5}{a_i}$$
    The type $A$ player wishes to move whenever: 
    $$ \frac{1.5}{a_i+1} + \costun > \frac{1.5}{a_i} $$
    $$ \costun > 1.5 \cd \frac{1}{a_i \cd (a_i+1)} \geq 1.5 \cd \frac{1}{2 \cd 3} = 0.25 $$
    Note that the lefthand side is decreasing in $a_i$, which is where the lower bound comes from. Note that if $a_i \geq 3$, the lower bound becomes $\costun \geq \frac{3}{2} \cd \frac{1}{3 \cd 4} = \frac{1}{8} = 0.125$. 
    \item Currently, the cost to players of type $B$ is: 
    $$\abs{\frac{-0.5 \cd a_i + 1}{a_i + 1} -1} + \costun $$
    $$= \abs{\frac{-0.5 \cd a_i + 1 - \cd a_i -1}{a_i + 1}} + \costun =  \frac{1.5 \cd a_i}{a_i+1} + \costun$$
    If one player of type $B$ goes to label the empty item, the cost becomes: 
    $$\abs{-0.5-1}  = 1.5$$
    The type $B$ player wishes to move whenever: 
    $$ \frac{1.5 \cd a_i}{a_i+1} + \costun > 1.5 $$
    $$ \costun > \frac{1.5}{a_i+1} \geq \frac{1.5}{3} = 0.5$$
    \item We are in the scenario where $j$ with $a_i \geq 2, b_i = 1$ (the sole type $B$ player is on item $i$). Suppose that there exists another item $j$ with $a_j = 1$. Then, the type $B$ player can reduce its cost by moving from item $i$ to item $j$ whenever: 
    $$\abs{\frac{-0.5 \cd a_i +1}{a_i+1} - 1} + \abs{-0.5 - 1} $$
    $$> \abs{-0.5 - 1} + \abs{\frac{-0.5 + 1}{2} - 1}$$
    Or: 
    $$\abs{\frac{-0.5 \cd a_i +1}{a_i+1} - 1}  > \abs{\frac{-0.5 + 1}{2} - 1}$$
    For $a_i\geq 2$, the term inside the absolute value is given by: 
    $$\abs{\frac{-0.5 \cd a_i +1}{a_i+1} - 1} = \frac{0.5 \cd a_i -1}{a_i +1}$$
    which is increasing in $a_i$, meaning that the inequality is always satisfied. 
    \end{itemize}
\end{itemize}

Now, we consider all arrangements where there are no more than 2 players labeling an item. We similarly know that if there are 2 players of bias $-0.5$ labeling an item, then they can reduce their cost by having one leave to label an empty item (maintaining the same mean and eliminating empty cost). Therefore, we can focus on how the player with bias $1$ is arranged. We will find that considering only the first 3 items will be sufficient. For conciseness, we will use the notation $\{1\}$ $\{-0.5\}, \{-0.5\}$ to mean that three players are all labeling a single item each (one player with bias 1, two players with bias $-0.5$). 

\begin{itemize}
    \item $\{1\}$ $\{-0.5\}, \{-0.5\}$ goes to $\{1, -0.5\}, \{\}, \{-0.5\}$. Originally, the -0.5 players have cost $1.5 + 0 + 0 = 1.5$. The $-0.5$ player wishes to move, which produces means $0.25, \emptyset, -0.5$. (The second item has undefined mean, given that it has no labels). This gives the $-0.5$ players cost $0.75 + \costun + 0$. This is lower if $0.75 + \costun < 1.5$, which is satisfied for $\costun < 0.75$. 
    \item $\{1, -0.5\}, \{\}, \{-0.5\}$ goes to $\{1, -0.5, -0.5\}, \{\}, \{\}$. The $-0.5$ player wishes to move again. Before, the means are $0.25, \emptyset, -0.5$, giving it a cost of $0.75 + \costun + 0$. After it moves, the means are $0, \emptyset, \emptyset$, which gives it a cost of $0.5+ 2 \cd \costun$. This is lower if $0.5 + 2 \cd \costun < 0.75 + \costun$, or  $\costun < 0.25$. 
    \item $\{1, -0.5, -0.5\}, \{\}, \{\}$ goes to $\{-0.5, -0.5\}, \{1\}, \{\}$. The 1 player wishes to move. Before, the means are $-0.5, \emptyset, \emptyset$, which gives it a cost of $1.5 + \costun \cd 2$. After it moves, the means are $-0.5, 1, \text{n/a}$, which gives it a cost of $1.5 + 0 + \costun < 1.5 + 2 \cd \costun$, which is always satisfied.  
    \item $\{-0.5, -0.5\}, \{1\}, \{\}$ goes to $\{-0.5,\}, \{1\}, \{-0.5\}$. The $-0.5$ player wishes to move. Before, the means are $-0.5, 1,\emptyset$, which gives it cost $0 + 1.5 + \costun$. After it moves, the medians are $-0.5, 1, -0.5$, which gives it cost $0 + 1.5 + 0 = 1.5 < 1.5 + \costun$, which is always satisfied. 
\end{itemize}
These cases cover all possible arrangements, showing that there always exists a agent that wishes to move. 
\end{proof}

\meanstablen*
\begin{proof}
Because we have assumed $\nplayer_a \geq \nplayer_b$, the only possible cases are $\nplayer_a =3, \nplayer_b=0$ (where there is always a stable arrangement with each player on a separate item) and $\nplayer_a=2, \nplayer_b=1$, which the remainder of this proof will analyze.

There are four possible arrangements of 3 agents where two of them have the same bias:
$$\{a\}, \{ a\}, \{b\} \quad \{a, a\}, \{b\} \quad \{a, b\}, \{a\} \quad \{a, a, c\}$$
We will use the notation 
$$\{a, b\}, \{\}\rightarrow_a\{a\}, \{b\}$$
to mean that a player of type $A$ gets strictly lower cost when it leaves an item it is competing in with type $B$ to compete in an empty item. We will say $\{a, b\}, \{\}\not \rightarrow_a\{a\}, \{b\}$
when the player of type $A$ does not get strictly lower cost in $\{a\}, \{b\}$ and say $\{a\}, \{b\} \rightarrow_a  \{a, b\}, \{\}$ when the player of type $A$ gets strictly lower cost in $\{a, b\}$. 
Often, to be concise, we will drop $\{\}$ terms and simply write $\{a, b\} \rightarrow_a \{a\}, \{b\}$
as the second item is left empty on the lefthand side. 

First, we claim that: 
\begin{equation}\label{eq:equivalent}
\{a, b\} \rightarrow_a \{a\}, \{b\} \quad \Leftrightarrow\quad \{a, b\} \rightarrow_b \{a\}, \{b\} 
\end{equation}
This is proved in the process of proving Lemma \ref{lem:twostable}, which showed that two players simultaneously either wish to be labeling the same item or labeling separate items, and the condition for when they wish to label separate items is whenever $\costun > 0.5 \cd \abs{\beta_b - \beta_a}$. 
Next, we know immediately that: 
$$\{a, a\} \{b\} \rightarrow_{a} \{a\}, \{a\}, \{b\}$$
and 
$$\{a, a\} \{b\} \rightarrow \{a, b\}, \{a\}$$
Because both players of type $a$ have exactly the same type, they can always reduce their cost by either leaving to label another item, or moving to the same item as type $b$. 

Next, we derive conditions for when different players would prefer to leave different arrangements. Specifically, 
$$\{a, a, b\} \rightarrow_a \{a, b\}, \{a\} $$
occurs when: 
$$\frac{1}{3} \cd \{\beta_b - \beta_a\} + \costun > 0.5 \cd \abs{\beta_a - \beta_b} + 0 $$
$$c > \frac{1}{6} \cd \abs{\bias_a - \bias_b} $$
Additionally, we can show that: 
$$\{a, a, b\} \rightarrow_b \{a, a\}, \{b\}$$
occurs whenever: 
$$\frac{2}{3} \cd \abs{\bias_a - \bias_b} + \costun > \abs{\bias_a - \bias_b} + 0$$
$$\costun > \frac{1}{3} \cd \abs{\bias_a - \bias_b}$$

Next, we can analyze some cases: \\
\textbf{Case 1:} $\costun \in \left(0, \frac{1}{6} \cd \abs{\bias_a - \bias_b}\right]$: \\
If this is the case, then $\{a, a, b\}$ is stable: from our prior analysis, we know that neither players of type $a$ or $b$ can reduce their cost by leaving, so having all players be together is stable. 

\noindent \textbf{Case 2:} $\costun \in \left(\frac{1}{6} \cd \abs{\bias_a - \bias_b}, 0.5 \cd \abs{\bias_a - \bias_b}\right]$:\\
If this is the case, then $\{a, b\}, \{a\}$ is stable. From our prior analysis, we know $\costun$ is high enough that $\{a, a, b\} \rightarrow_{a} \{a, b\}, \{a\}$. However, it is also low enough that $\{a\}, \{a\}, \{b\} \rightarrow_{a, b} \{a, b\} \{a\}$. Therefore, $\{a, b\}, \{a\}$ is stable. 

\noindent  \textbf{Case 3:} $\costun  > 0.5 \cd \abs{\bias_a - \bias_b}$\\
In this case, $\costun$ is high enough that $\{a\} \{a\} \{b\}$ is stable: both players of type $a$ and $b$ prefer it to being together. 

Taken together, these three cases cover all possible settings, and show that a stable arrangement always exists. 

\end{proof}

\cuNsmallNE*
\begin{proof}
First, we can immediately see that no player $i$ would wish to label an item $j$ that is currently empty. In doing so, they would simply be moving from item $i$ to item $j$ and keeping their cost exactly the same. 

The other case is showing that no player labeling item $i$ would wish to label an item $j$ that already has a label on it. In doing so, that player would be leaving item $i$ empty: from Lemma \ref{lem:nounlabeled}, we know that this would increase player costs exactly whenever $\costun \geq 0.5 \cd \abs{\bias_a - \beta_{b}}$. 
\end{proof}

\section{Proofs for Appendix \ref{app:unevenweight}}\label{app:proofsunevenweight}

\nounlabeledw*
\begin{proof}
In other words, we want to ensure that for any player $\type$, the cost of competing in any item $i$ (leaving 
any other item $j$ empty) is higher than the cost of leaving item $i$ to competing in item $j$ alone:  

$$\weight_i \cd \abs{f(S_i) - \bias_{\type}} + \weight_j \cd \costun \geq \weight_i \cd \abs{f(S_i \setminus \bias_{\type})}+ \weight_j \cd 0$$

First, we will analyze the case with mean outcome function. For an agent of type $A$, the cost it experiences from an item with $a$ agents of type $A$ and $b$ agents of type $B$ is given by: 
$$\abs{\frac{a \cd \bias_a + b \cd \bias_b}{a + b} - \bias_a} = \abs{\frac{a \cd \bias_a + b \cd \bias_b - (a+b) \cd \bias_a}{a+b}}$$
$$= \frac{b}{a+b} \cd \abs{\bias_a - \bias_b} $$
By identical reasoning, the cost to an agent of type $B$ is: 
$$ \frac{a}{a+b} \cd \abs{\bias_a - \bias_b} $$
Note that this construction immediately tells us that agent strategy must be independent of biases. For every item with $a+b>0$, an agent's cost is solely a function of $a$ and $b$, scaled by a constant factor of $\abs{\bias_a - \bias_b}$. 

Next, we will work on determining $\costun$ so that no item will ever be left empty. Again, we wish to show that: 
$$\weight_i \cd \abs{f(S_i) - \bias_{\type}} + \weight_j \cd \costun \geq \weight_i \cd \abs{f(S_i \setminus \bias_{\type})} + \weight_j \cd 0$$
The worst-case scenario (where this inequality is hardest to satisfy) occurs when the empty item has very low weight, relative to the item it is competing in ($\weight_i >> \weight_j$). 
If we consider a reference player of type $A$, with $a$ players of type $A$ on item $i$ and $b$ of type $B$, then this becomes: 
$$\weight_i \cd \abs{\beta_a - \beta_b} \cd \frac{b}{a+b}+ \weight_j \cd \costun \geq \weight_i \cd \abs{\beta_a - \beta_b} \cd \frac{b}{a+b-1} + 0$$
$$\weight_j \cd \costun \geq \weight_i \cd \abs{\beta_a - \beta_b} \cd b \cd \p{ \frac{1}{(a+b-1)} -\frac{1}{(a+b)}} $$

$$\weight_j \cd \costun \geq \weight_i \cd \abs{\beta_a - \beta_b} \cd b \cd  \frac{b}{(a+b-1)\cd (a+b)} $$

Next, we'll upper bound the term on the RHS. The RHS shrinks with $a$, so we can lower bound this by setting $a =1$. We know that $a\geq 1$ because we have assumed there is at least one player of type $A$ that wishes to move from the given item. The condition simplifies to: 
$$\costun \cd \weight_j \geq \weight_i \cd  \abs{\beta_a - \beta_b}  \cd \frac{1}{1+b}$$
We similarly must have $b\geq 1$ (or else we're just modeling a single player of type $A$ move from one item to another). If we set $b=1$, then this goes to 1/2, which gives the desired condition. Intuitively, this tells us that we need that the cost of leaving something unlabeled is greater than half the distance between the two biases. 

Next, we will consider the case where the outcome function is equal to the median. Again, we wish to show that: 
$$\weight_i \cd \abs{f(S_i) - \bias_{\type}} + \weight_j \cd \costun \geq \weight_i \cd \abs{f(S_i \setminus \bias_{\type})} + \weight_j \cd 0$$
We will analyze multiple different cases for the potential outcome functions $f(S_i)$ and $\abs{f(S_i \setminus \bias_{\type})}$. Again, we will look from the perspective of a type $A$ agent on item $i$ considering moving to another item $j$ that is empty: 

\begin{itemize}
    \item $f(S_i) = \bias_a$ and $\abs{f(S_i \setminus \bias_{\type})} = \bias_a$. The inequality becomes: 
    $$\weight_i \cd 0 + \weight_j \cd \costun \geq \weight_i \cd 0 + \weight_j \cd 0$$
    which is satisfied automatically. 
    \item $f(S_i) = \bias_a$ and $\abs{f(S_i \setminus \bias_{\type})} = \frac{1}{2} \cd (\bias_a + \bias_b)$. The inequality becomes: 
    $$ \weight_i \cd 0 + \weight_j \cd \costun \geq 0.5 \cd \weight_i \cd \abs{\bias_a - \bias_b} + \weight_j \cd 0$$
    $$\costun \geq 0.5 \cd \frac{\weight_i}{\weight_j} \cd  \abs{\bias_a - \bias_b}$$
    \item $f(S_i) = \frac{1}{2} \cd (\bias_a + \bias_b) $ and $\abs{f(S_i \setminus \bias_{\type})} = \bias_b$. The inequality becomes: 
    $$0.5 \cd \weight_i \cd  \abs{\bias_a - \bias_b}  + \weight_j \cd \costun \geq \weight_i \cd  \abs{\bias_a - \bias_b}+ \weight_j \cd 0$$
    $$\costun \geq 0.5 \cd \frac{\weight_i}{\weight_j} \cd \abs{\beta_a - \beta_b}$$
    \item $f(S_i) = \frac{1}{2} \cd \bias_b$ and $\abs{f(S_i \setminus \bias_{\type})} = \bias_b$. The inequality becomes: 
    $$\weight_i \cd \abs{\bias_a - \bias_b} + \weight_j \cd \costun \geq \weight_i \cd \abs{\bias_a - \bias_b} + \weight_j \cd 0$$
    which is always satisfied. 
\end{itemize}
The only inequality that isn't automatically satisfied is $\costun \geq 0.5 \cd \frac{\weight_i}{\weight_j} \cd  \abs{\bias_a - \bias_b}$, which is the same inequality as for the mean outcome function, and satisfied by the same reasoning. 

Finally, we will show that agents' incentives are independent of biases $\bias_a, \bias_b$.

This proof comes almost immediately. \\
For mean outcome function, we can immediately see from the agent cost that agent strategy must be independent of biases. For every item with $a+b>0$, an agent's cost is solely a function of $a$ and $b$, scaled by a constant factor of $\abs{\bias_a - \bias_b}$. 

For median allocation, for any agent with bias $\bias_a$ and any set $S_i$, the outcome function has three possible values: $\bias_a$ (giving cost to agent $a$ of 0), $0.5 \cd (\bias_a + \bias_b)$ in the event of ties (giving cost to agent $a$ of $0.5 \cd \abs{\bias_a - \bias_b}$), or $\bias_b$ (giving cost to agent $a$ of $\abs{\bias_a - \bias_b}$). All of these are simply scaled values of $\abs{\bias_a - \bias_b}$, which means incentives are independent of the values $\bias_a, \bias_b$. 
\end{proof}

\exmednotstab*

\begin{proof}
In this proof, we will use the notation
$$\{a, b\}, \{a, b\}, \{a\},\{b\}$$
 to illustrate that there are 4 items in total, two of which with exactly 2 players on it (one of each type), and two items with exactly one player (one item with a single type $A$ player, and one item with a single type $B$ player). 
Lemma \ref{lem:medianstableequal} would suggest that the arrangement 
$$\{a, b\}, \{a, b\}, \{a\},\{b\}$$
would be stable. If $w_1 = w_2$, then this would be satisfied: none of the singleton players could move, and none of the $\{a, b\}$ players wish to move. If a player of type $A$from item 2 moved to label item 1, they would go from experiencing cost 
$$0.5 \cd \weight_1  + 0.5 \cd \weight_2 + \weight_3$$
to experiencing cost: 
$$\weight_2 + \weight_3$$
which is identical when $\weight_1=\weight_2$. However, when $\weight_1 > \weight_2$, then this move does reduce cost, meaning that the original arrangement was unstable.

We will further show that no possible arrangement is stable. Note that items $2, 3, 4$ have identical weight and are therefore interchangable. Because $\nplayer_a + \nplayer_b = 6$ and $\nitem = 4$, if no item can be left empty, then no item can have more than 3 players. We will consider each case based on the maximum number of agents on an item.

\begin{itemize}
    \item[Case 1] No more than 3 agents on one item:  The arrangement $\{a\}, \{b,b, b\}, \{a\}, \{a\}$ goes to $\{a\}, \{b, b\}, \{a, b\}, \{a\}$ because type $B$'s cost goes from $\weight_1 + \weight_3 + \weight_4$ to $\weight_1 + 0.5 \cd \weight_3 + \weight_4$. (Type $B$ continues to win on item 2 and now also ties on item 3). Identical reasoning holds for any symmetrical case with 3 players of the same type on any item. 
    \item[Case 2] No more than 3 agents on one item: The arrangement $\{a, a, b\}, \{b\}, \{a\},\{b\}$ moves to $\{a, a\}, \{b\}, \{a, b\},\{b\}$: the type $B$ player on item 1 can move to compete in item 3, which takes its cost from $\weight_1 + \weight_3$ to $\weight_1 + 0.5 \cd \weight_3$. (Type $B$ continues to lose item 1 and now ties on item 3). Identical reasoning holds for the symmetric case with type $A$, as well as for arrangements where player $i>1$ has 3 players: for example, the arrangement $\{a\}, \{a, a, b\}, \{a\},\{b\}$ goes to $\{a\}, \{a, a\}, \{a, b\},\{b\}$ because the type $B$ players cost goes from $\weight_1 + \weight_2 + \weight_3$ to $\weight_1 + \weight_2 + 0.5 \cd \weight_3$. 
    \item[Case 3] No more than 2 agents on an item: Given the arrangement $\{a, b\}, \{a, b\}, \{a\},\{b\}$, as described above, either type $A$ or type $B$ players from item 1 would wish to move to compete in item 1. This is the case analyzed at the beginning of the proof, where we showed this is unstable for $\weight_1 > \weight_2$. 
    \item[Case 4] No more than 2 agents on an item: The arrangement $\{a, a\}, \{b\}, \{a, b\},\{b\}$ goes to $\{a\}, \{a, b\}, \{a, b\},\{b\}$ because player type $A$'s cost goes from $\weight_2 + 0.5 \cd \weight_3 + \weight_4$ to $0.5 \cd (\weight_2 + \weight_3) + \weight_4$. Similarly, any arrangement with exactly 2 agent of the same type on an item must leave at least one other item with exactly 1 agent of the opposite type, and will be unstable for the same reasons. 
\end{itemize}

This description is exhaustive: If there is a maximum of 3 agents on an item, then they must all be of the same type (Case 1) or two of the same type, and one of another type (Case 2). If there is a maximum of 2 agents on an item, then if one item has exactly one agent of each type, then another item must have exactly one agent of each type (Case 3) or must have exactly 2 agents of one type (Case 4). 
\end{proof}

\medianstablerelaxequal*

\begin{proof}
This proof is very similar to that of Lemma \ref{lem:medianstableequal}, so we will simply note key differences in the analysis. Throughout, we will assume that $\weight_i \geq \weight_j$ for $i < j$ (the items are organized in descending order of weight). 

First, we will suppose that either: 
$$\nplayer_a + \nplayer_b \geq 2 \cd \nitem +1 \text{ or } \nplayer_a + \nplayer_b = 2 \cd \nitem \text{ with } \nplayer_a, \nplayer_b \text{ even } $$

Again, we will allocate agents according to Algorithm \ref{alg:allocatemore}. This arrangement is stable by identical reasoning: for each item, each type wins by at least 2, so no single agent acting alone can change the outcome, regardless of weights. 

\begin{itemize}
    \item For every item where there is more than 1 labeler, type $A$ and type $B$ tie exactly. 
    \item Every other item has exactly one labeler, which can be either type $A$ or type $B$. 
    \item Any item where players tie has higher or equal weight to any item where a single player wins.
\end{itemize}

Next, we consider the case where $\nplayer_a + \nplayer_b \leq 2 \cd \nitem$. We will show constructively that it is possible to create an arrangement satisfying the following criteria: 
\begin{itemize}
    \item For every item where there is more than 1 agent, type $A$ and type $B$ tie exactly. 
    \item Every other item has exactly one agent, which can be either type $A$ or type $B$.
    \item Any item where players tie has higher or equal weight to any item where a single player wins.
\end{itemize}

This type of construction is stable by the following reasoning: 
\begin{itemize}
        \item None of the single agents can move (they can't leave an item empty). 
        \item No agent on an item with multiple agents wishes to leave - they would go from winning a single item and losing another, to losing on that item and tying on another, which gives equal costs when weights are equal. 
        \item No agent on an item with multiple agents wishes to leave: they would go from tying on item $i$ and losing on item $j$ to losing on item $i$ and tying on item $j$. This gives higher or equal cost when $w_i \geq w_j$, which is satisfied by construction. 
\end{itemize}

Note that we require $\weight_1 = \weight_2$ by the example given in Lemma \ref{lem:exmednotstab}. 
\end{proof}

\twostablew*
\begin{proof} 
WLOG, assume that the items have weights in descending order, so $\weight_1 \geq \weight_2 \geq \ldots \weight_{\nitem}$. We will consider two players with bias $\beta_a, \beta_b$ and show that a stable arrangement always involves both players on item 1 or them split over item 1 and item 2. 

We will first characterize players' incentives to label different items. 

Player with bias $\bias_a$ prefers labeling item 1 with the other player of $\bias_b$ (as opposed to leaving to label item 2) whenever: 
\begin{align*}
\weight_1 \cd \abs{f(\{\beta_a, \beta_b\}) - \beta_a} +\weight_2 \cd \costun & \leq \weight_1 \cd \abs{f(\beta_b) - \beta_a} + \weight_2 \cd 0 \\
\weight_1 \cd \abs{0.5 \cd \p{\beta_a + \beta_b} - \beta_a} +\weight_2 \cd \costun & \leq  \weight_1 \cd \abs{\beta_b - \beta_a}\\
\weight_1 \cd 0.5 \cd \abs{\beta_a- \beta_b} +\weight_2 \cd \costun & \leq  \weight_1 \cd \abs{\beta_b - \beta_a}\\
  \weight_2 \cd \costun & \leq 0.5 \cd \weight_1 \cd  \abs{\beta_b - \beta_a}
\end{align*}
By identical reasoning, the player of bias $\bias_b$ prefers labeling item 1 with the player of bias $\bias_a$ (as opposed to leaving to label item 2) exactly whenever: 
$$\weight_2 \cd \costun \leq 0.5 \cd \weight_1 \cd \abs{\bias_b - \bias_a}$$
These terms are exactly equivalent, so when $\weight_2 \cd \costun \leq 0.5 \cd \weight_1 \cd \abs{\bias_b - \bias_a}$, having both players be labeling item 1 is stable: no player wishes to leave to label item 2, and because $\weight_2 \geq \weight_i \ \forall i \ne 1$, they also don't want to wish to leave to label any other item. 

Conversely, when $\weight_2 \cd \costun > 0.5 \cd \weight_1 \cd \abs{\bias_b - \bias_a}$, then it is stable to have 1 player on item 1 and 1 player on item 2. The player on item 2 doesn't wish to leave it to label item 1 (by our setting of the parameters). Similarly, we know that the player on item 1 doesn't wish to go label item 2, because it would prefer being separate to labeling item 1 together, and $\weight_2 \leq \weight_1$ so its preference for labeling item 2 is lower than its preference for labeling item 1. Finally, we know that neither player would prefer to go label another item, because $\weight_i\leq \weight_1, \weight_2 \ \forall i>2$. 
\end{proof}

\meanbigunstablew*
\begin{proof}
Given $\nplayer_a = 1, \nplayer_b=1$ and cost $\costun$ satisfying Lemma \ref{lem:nounlabeledw}, we can always find a stable arrangement by putting exactly one item on each of the fronts (no agent will wish to change which item they are labeling because doing so would require leaving an item unlabeled). Therefore, the rest of the proof will handle the case where $\nplayer_a + \nplayer_b >2$. 

Specifically, we will set weights $\weight_1, \weight_2$ such that the following arrangement will be stable: having the first item have $(\nplayer_a-1, \nplayer_b)$ agents of type $A$ and $B$ respectively, and the second item have $(1, 0)$ agents of type $A$ and $B$ respectively. We will next show that this arrangement is stable. 

First, we know that the single player of type $A$ on item 2 never wishes to leave and label item 1: doing so would involve leaving item 2 unlabeled, which by Lemma \ref{lem:nounlabeledw} no agent wishes to do. Additionally, we know that no agent of type $A$ on item 1 wishes to leave to label item 2: doing so wouldn't change the outcome on item 2 (which is already solely made up of agents of type $A$) and would strictly increase the cost on item 1 (because fewer agents of type $A$ would be on item 1). Therefore, the only non-trivial case is considering whether agents of type $B$ would prefer to leave item 1 and label item 2 instead. 

An agent of type $B$ would \emph{not} prefer to leave item 1 when doing so would increase their cost: 
$$\frac{\nplayer_a-1}{\nplayer_a-1 + \nplayer_b} \cd \weight_1 + 1 \cd \weight_2 < \frac{\nplayer_a-1}{\nplayer_a + \nplayer_b-2}\cd \weight_1 + 0.5 \cd \weight_2$$
$$0.5 \cd \weight_2 < \weight_1 \cd (\nplayer_a-1) \cd \p{\frac{1}{\nplayer_a + \nplayer_b-2} - \frac{1}{\nplayer_a + \nplayer_b-1}}$$
$$\frac{\weight_2}{\weight_1} < 2 \cd \frac{\nplayer_a-1}{(\nplayer_a + \nplayer_b-2) \cd (\nplayer_a + \nplayer_b-1)}$$
Because $\nplayer_a + \nplayer_b >2$ in this case, the denominator is strictly greater than 0. Thus, setting the weights $\weight_1, \weight_2$ such that $\weight_1 + \weight_2 = 1$ and $\frac{\weight_2}{\weight_1}$ satisfies the above ratio ensures that no agent can strictly reduce their cost by changing which item they label. 
\end{proof}

\meanbowlw*

\begin{proof}
In order to show this, we will look at a relaxed (continuous) version of this problem. In this relaxed version of the problem, we will assume that, instead of agents coming in integer units, they can be allocated fractionally across items.  

The cost to player of type $A$ is given by 
$\sum_{i=1}^{\nitem}\weight_i \cd \frac{b_i}{a_i+b_i}$. 
The partial derivative of this with respect to $a_i$ is: 
$$\frac{\partial}{\partial a_i} \sum_{i=1}^{\nitem}\weight_i \cd \frac{b_i}{a_i+b_i} = -\weight_i \cd \frac{b_i}{(a_i + b_i)^2}$$
In order for us to be at a stable point, we need that the derivative wrt $a_i$ must be equal to the derivative wrt $a_j$ for any $i \ne j$ and $\sum_{i=1}^{\nitem} a_i = \nplayer_a$ (identical criteria for $B$). If the first is not satisfied, then type $A$ could strictly reduce its cost by changing its allocations between items, and if the second is not satisfied, then type $A$ could again reduce its cost by allocating more agents onto items. We can enforce both of these by setting: 
$$\weight_i \cd \frac{b_i}{(a_i + b_i)^2}= \weight_{\nitem} \cd \frac{\nplayer_b - \sum_{i \ne \nitem}b_i}{(\nplayer_a - \sum_{i \ne \nitem}a_i + \nplayer_b - \sum_{i\ne \nitem} b_i)^2}$$
$$\frac{b_i}{\nplayer_b - \sum_{i \ne \nitem}b_i}= \frac{\weight_{\nitem}}{\weight_i} \cd \frac{(a_i + b_i)^2}{(\nplayer_a - \sum_{i \ne \nitem }a_i + \nplayer_b - \sum_{i\ne \nitem} b_i)^2}$$
By identical reasoning, this means that for player $B$, the derivative of its cost wrt $b_i$ is:  
$$\frac{a_i}{\nplayer_a - \sum_{i \ne \nitem}a_i}= \frac{\weight_{\nitem}}{\weight_i} \cd \frac{(a_i + b_i)^2}{(\nplayer_a - \sum_{i \ne \nitem }a_i + \nplayer_b - \sum_{i\ne \nitem} b_i)^2}$$
Setting these equal to each other gives: 
$$\frac{a_i}{\nplayer_a - \sum_{i \ne \nitem}a_i} = \frac{b_i}{\nplayer_b - \sum_{i \ne \nitem}b_i}$$
$$\nplayer_b - \sum_{i \ne \nitem}b_i = \frac{b_i}{a_i} \cd \p{\nplayer_a - \sum_{i \ne \nitem}a_i}$$
Substituting in for the overall derivative gives us that: 
\begin{align*}
 & \ \frac{a_i}{\nplayer_a - \sum_{i \ne \nitem}a_i}\\  & = \frac{\weight_{\nitem}}{\weight_i} \cd \frac{(a_i + b_i)^2}{(\nplayer_a - \sum_{i \ne \nitem }a_i + \nplayer_b - \sum_{i\ne \nitem} b_i)^2}\\
& = \frac{\weight_{\nitem}}{\weight_i} \cd \frac{(a_i + b_i)^2}{\p{\nplayer_a - \sum_{i \ne \nitem }a_i + \frac{b_i}{a_i} \cd \p{\nplayer_a - \sum_{i \ne \nitem}a_i}}^2}\\
& = \frac{\weight_{\nitem}}{\weight_i} \cd \frac{(a_i + b_i)^2}{\p{\nplayer_a - \sum_{i \ne \nitem}a_i}^2 \cd \p{1 + \frac{b_i}{a_i}}^2}\\
& = \frac{\weight_{\nitem}}{\weight_i} \cd \frac{a_i^2 \cd (a_i + b_i)^2}{\p{\nplayer_a - \sum_{i \ne \nitem}a_i}^2 \cd \p{a_i + b_i}^2}\\
& = \frac{\weight_{\nitem}}{\weight_i} \cd \frac{a_i^2 }{\p{\nplayer_a - \sum_{i \ne \nitem}a_i}^2}\\
\end{align*}
Cancelling common factors on each side of the equality gives: 
$$\frac{\nplayer_a - \sum_{i \ne \nitem}a_i}{\weight_{\nitem}} = \frac{a_i}{\weight_i}$$
Or $\frac{a_{\nitem}}{\weight_{\nitem}} \cd \weight_i = a_i$. 
We can apply the equality $\sum_{i=0}^{\nitem} a_i = \nplayer_a$ to obtain: 
$$\sum_{i=0}^{\nitem}\frac{a_{\nitem}}{\weight_{\nitem}} \cd \weight_i = \nplayer_a \quad \Rightarrow \quad a_{\nitem} = \weight_{\nitem} \cd \nplayer_a$$
which implies $a_i = \weight_i \cd \frac{\weight_{\nitem} \cd \nplayer_a}{\weight_{\nitem}} = \weight_i \cd \nplayer_a$
and $b_i  = \nplayer_b \cd \weight_i$ symmetrically. This tells us that the only possible equilibrium is when both $A$ and $B$ players have set $a_i, b_i$ exactly proportional to the weight each item has.
\end{proof}

\meanstablenw*
\begin{proof}
There are five possible arrangements of 3 agents:
$$\{a\}, \{ b\}, \{c\} \quad \{a, b\}, \{c\} \quad \{a, c\}, \{b\} \quad \{b, c\}, \{a\} \quad \{a, b, c\}$$
We will use the notation 
$$\{a, b\}, \{\}\rightarrow_a\{a\}, \{b\}$$
to mean that a player of type $A$ gets strictly lower cost when it leaves an item it is competing in with type $B$ to compete in an empty item. We will say $\{a, b\}, \{\}\not \rightarrow_a\{a\}, \{b\}$
when the player of type $A$ does not get strictly lower cost in $\{a\}, \{b\}$ and say $\{a\}, \{b\} \rightarrow_a  \{a, b\}, \{\}$ when the player of type $A$ gets strictly lower cost in $\{a, b\}$. 
Often, to be concise, we will drop $\{\}$ terms and simply write $\{a, b\} \rightarrow_a \{a\}, \{b\}$
as the second item is left empty on the lefthand side. 

First, we claim that: 
\begin{equation}\label{eq:equivalent}
\{a, b\} \rightarrow_a \{a\}, \{b\} \quad \Leftrightarrow\quad \{a, b\} \rightarrow_b \{a\}, \{b\} 
\end{equation}
Note that the lefthand term of Equation \ref{eq:equivalent} is satisfied whenever: 
$$\abs{f(\{\beta_a, \beta_b\}) - \beta_a} + \costun \leq 0 + \abs{f(\beta_b) - \beta_a}$$
$$0.5 \cd \abs{\beta_b - \beta_a} + \costun \leq \abs{\beta_b - \beta_a}$$
The righthand term of Equation \ref{eq:equivalent} is satisfied whenever: 
$$\abs{f(\{\beta_a, \beta_b\}) - \beta_b} + \costun \leq \abs{f(\beta_a) - \beta_b} + 0$$
$$0.5 \cd \abs{\beta_b - \beta_a} + \costun \leq \abs{\beta_b - \beta_a}$$
These terms are exactly equivalent, showing Equation \ref{eq:equivalent} is satisfied. We will use this result in the following analysis. 

Next, we will show that if any two pairs of agents prefer competing in the same item (as opposed to competing in different items), then it is the two agents with most dissimilar biases (e.g. $i, j$ given $\max_{i, j \in \{a, b, c\}} \abs{\beta_i - \beta_j}$). In order to show this, we will write out the cost that an agent of type $A$ gets for every possible arrangement. Because agents are interchangable, this gives costs for every other agent in different arrangements, up to relabeling. 

\begin{align*}\label{eq:costs3}
 & \{a\}, \{b\}, \{c\}    \quad   & \abs{\beta_b - \beta_a} + \abs{\beta_c - \beta_a}\\
& \{a, b\}, \{c\}   \quad  & 0.5 \cd \abs{\beta_b - \beta_a} + \abs{\beta_c - \beta_a} + \costun\\
 & \{a, c\}, \{b\}   \quad &  0.5 \cd \abs{\beta_a - \beta_c} + \abs{\beta_b - \beta_a} + \costun\\
& \{b, c\}, \{a\}  \quad  & 0.5 \cd \abs{\beta_b+ \beta_c - 2 \cd \beta_a} + \costun\\
& \{a, b, c\}  \quad  & \frac{1}{3} \cd \abs{\beta_b+ \beta_c - 2 \cd \beta_a} + 2 \cd \costun
\end{align*}
The second and third lines in the equations above show us that 
$$\{a, b\}, \{c\}  \rightarrow_a \{a, c\}, \{b\}$$
exactly whenever: 
$$0.5 \cd \abs{\beta_b - \beta_a} + \abs{\beta_c - \beta_a} + \costun >  0.5 \cd \abs{\beta_a - \beta_c} + \abs{\beta_c - \beta_a} + \costun$$
$$ \abs{\beta_c - \beta_a}  > \abs{\beta_b - \beta_a} $$
which is whenever $\beta_c$ is further from $\beta_a$ than $\beta_b$ is. In this analysis, WLOG we will say that $\beta_b, \beta_c$ are the most dissimilar (because agents are interchangable, this is true up to relabeling). Note that this implies that $\{a, c\}, \{b\}$ and $\{a, b\}, \{c\}$ can never be stable, because both $b, c$ would prefer being together to being with $a$. 

Next, we can analyze some cases: \\
\textbf{Case 1:}\\
If the $\{b, c\}, \{a\} \rightarrow_{b, c} \{a\}, \{b\}, \{c\}$, then $\{a\}, \{b\}, \{c\}$ is stable. \\
We know from our prior reasoning that if the two most dissimlar players do not wish to compete in the same item, then no other pair of players do. This means that the arrangement with one agent per item is stable.

\textbf{Case 2:}\\
If: 
$$\{a\}, \{b\}, \{c\} \rightarrow_{b, c} \{b, c\}, \{a\} \text{ and } \{b, c\}, \{a\} \not \rightarrow_{a} \{a, b, c\}$$
then $\{b, c\}, \{a\}$ is stable. This is because the $\{b, c\}$ pair doesn't wish to split up and the type $A$ player doesn't wish to join. 

\textbf{Case 3:}\\
If no player wishes to leave $\{a, b, c\}$ to compete in an empty item, then $\{a, b, c\}$ is stable. \\
This is true simply by the statement: if no player wishes to move, then this arrangement must be stable. 

\textbf{Case 4:}\\
In this case, we will assume that: 
$$\{a\}, \{b\}, \{c\} \rightarrow_{b, c} \{b, c\}, \{a\} \text{ and } \{b, c\}, \{a\} \rightarrow_{a} \{a, b, c\}$$
We will also assume that at least one of $b, c$ wishes to leave $\{a, b, c\}$. WLOG, we will assume it is $c$ (up to relabeling). 
Note that this implies none of these arrangements can be stable: 
$$\{a\}, \{b\}, \{c\} \quad \{b, c\}, \{a\} \quad \{a, b, c\}$$
Additionally, by our prior reasoning, if players $b, c$ prefer competing together to separate, then we know that no other pair can be stable: 
$$\{a, b\}, \{c\} \quad \{a, c\}, \{b\}$$
The cycle that is given by this set of patterns is given by: 
$$\{b, c\}, \{a\} \rightarrow_a \{a, b, c\} \rightarrow_c \{a, b\}, \{c\} \rightarrow_{b, c} \{a\}, \{b, c\}$$
We will show that this type of cycle cannot exist. 

Note that the first move occurs when: 
$$\abs{\frac{\beta_b + \beta_c}{2} - \beta_a} > \abs{\frac{\beta_b + \beta_c + \beta_a}{3} - \beta_a} + \costun$$
$$\abs{\frac{\beta_b + \beta_c - 2\beta_a}{2}} > \abs{\frac{\beta_b + \beta_c + \beta_a - 3 \beta_a}{3}} + \costun$$
$$\frac{1}{6} \cd \abs{\beta_b + \beta_c - 2 \beta_a} > \costun$$
The second move occurs when:
$$\abs{\frac{\beta_a + \beta_b + \beta_c}{3} - \beta_c} + \costun > \abs{\frac{\beta_a + \beta_b}{2} - \beta_c}$$
$$\abs{\frac{\beta_a + \beta_b -2 \beta_c}{3}} + \costun > \abs{\frac{\beta_a + \beta_b - 2 \beta_c}{2}}$$
$$\costun > \frac{1}{6} \cd \abs{\beta_a + \beta_b - 2 \beta_c}$$
And the third move occurs when: 
$$0.5 \cd \abs{\beta_a - \beta_b} + \abs{\beta_c - \beta_b} > \abs{\beta_a - \beta_b} + 0.5 \cd \abs{\beta_c - \beta_b}$$
$$0.5 \cd \abs{\beta_b - \beta_c}> 0.5 \cd \abs{\beta_a - \beta_b}$$
Putting these together, this implies: 
\begin{equation}\label{eq:firsteqn4}
\abs{\beta_a + \beta_b - 2 \beta_c} < \abs{\beta_b + \beta_c - 2 \beta_a}
\end{equation}
\begin{equation}\label{eq:secondeqn4}
\abs{\beta_a - \beta_b}< \abs{\beta_b - \beta_c}
\end{equation}
We will show that these pairs of inequalities cannot be simultaneously satisfied. Set: 
$$x = \beta_b - \beta_a \quad y = \beta_c - \beta_b$$
Then, Equation \ref{eq:secondeqn4} becomes
$$\abs{-x}< \abs{-y}$$
$$\abs{x} < \abs{y}$$
For Equation \ref{eq:firsteqn4} note that the LHS is given by: 
$$\abs{\beta_a - \beta_c + \beta_b - \beta_c} = \abs{-x-y-y} = \abs{x + 2y}$$
where the equality is because: 
$$\abs{-\beta_b + \beta_a - 2\beta_c + 2 \beta_b} = \abs{\beta_a + \beta-b - 2 \beta_c}$$
as desired. The RHS is given by: 
$$\abs{2 \cd x + y}$$
Which is given by: 
$$\abs{2 \cd (\beta_b - \beta_a) + \beta_c - \beta_b} = \abs{\beta_b + \beta_c - 2 \cd \beta_a} $$
So together, Equations \ref{eq:firsteqn4} and \ref{eq:secondeqn4} become:

\begin{equation}\label{eq:firsteqn4p}
\abs{x + 2 \cd y} < \abs{2 \cd x + y}
\end{equation}
\begin{equation}\label{eq:secondeqn4p}
\abs{x} < \abs{y}
\end{equation}

Note that if both $x, y$ are positive, Eq. \ref{eq:firsteqn4p} becomes: 
$$x + 2 \cd y < 2 \cd x + y$$
$$y < x$$
which directly contradicts Eq. \ref{eq:secondeqn4p} ($\abs{x} = x < \abs{y} = y$). Similarly, if both $x, y$ are negative, then Eq. \ref{eq:firsteqn4p} becomes: 
$$-x - 2 \cd y < -2 \cd x - y$$
$$-y< -x$$
which contradicts Eq. \ref{eq:secondeqn4p} which has become
$$-x < -y$$
Finally, let's consider the case where one of $x, y$ is positive and one is negative. WLOG, say $x < 0 < y$. Then, we know that the LHS of Eq. \ref{eq:firsteqn4p} is positive (because $ \abs{x}< \abs{y}$ so $\abs{x} < 2\abs{y}$, which makes the overall $x + 2 \cd y$ positive). This means that the RHS of Eq. \ref{eq:firsteqn4p} must be positive. If $2 \cd x + y>0$, then this becomes: 
$$x + 2 \cd y < 2 \cd x + y$$
$$ y < x$$
which is violated because $x < 0 < y$. If $2 \cd x + y < 0$, then this becomes: 
$$x + 2 \cd y < -2 \cd x -y$$
$$3 \cd x + 3 \cd y < 0$$
which again violates $x < 0 < y$, with $\abs{x} < \abs{y}$ given by Eq. \ref{eq:secondeqn4p}. 

Taken together, these cases cover all possible combination of stable relations, showing that there can never be a cycle (and must always have a stable point). 
\end{proof}

\cuNsmallNEw*
\begin{proof}
For simplicity, we will say that item $i$ is labeled by player with bias $\beta_i$. In order for this arrangement to be stable, we need to know that no player $i$ wishes to label a different item $j$. 

First, we can immediately see that no player $i$ would wish to label an item $j$ that is currently empty. By the way we have arranged players, we know that $\weight_i \leq \weight_j$ for $i \geq j$, so labeling that item would have less impact on overall cost. 

Next, we need to show that no player $i$ wishes to label a different item $j$ that already has a label on it. 

No player $j \in [\nplayer]$ wishes to move to another item $k$ that already has a label so long as so long as:
 $$\sum_{i =1}^{\nplayer} \weight_i \cd \abs{\beta_i - \beta_j} + \costun \cd \sum_{i =\nplayer+1}^{\nitem} \weight_i $$
 $$\leq \sum_{i \in [\nplayer], i \ne k, j} \weight_i\cd \abs{\beta_i- \beta_j} + 0.5 \cd \weight_k \cd \abs{\beta_k- \beta_j} + \costun \cd \sum_{i=\nplayer+1}^{\nitem} \weight_i + \weight_j \cd \costun$$
 which implies: 
 $$\weight_k \cd \abs{\beta_k - \beta_j} \leq 0.5 \cd \weight_k \cd \abs{\beta_j- \beta_k}+ \weight_j \cd \costun$$
 $$0.5 \cd \weight_k \cd \abs{\beta_k - \beta_j} \leq  \weight_j \cd \costun $$
 $$0.5 \cd \weight_k \cd \abs{\beta_k - \beta_j}\leq \frac{\weight_j}{\weight_k} \cd \costun $$
 The RHS is lower bounded by setting $\weight_j = \weight_{\nitem}, \weight_k = \weight_0$. The LHS is upper bounded by plugging in for $\beta_0$ (the largest magnitude) and $\beta_{i^*}$ (the next largest magnitude with an opposite sign), meaning that the LHS will be equal to $\abs{\beta_0} + \abs{\beta_{i}}$ rather than $\abs{\beta_0} - \abs{\beta_{i}}$). This gives us: 
 $$0.5 \cd \abs{\beta_0 - \beta_{i*}}\leq \frac{\weight_\nitem}{\weight_0} \cd \costun$$
\end{proof}

\end{document}